\documentclass[a4paper,twocolumn,10pt]{article}
\usepackage[lmargin=1.5cm,rmargin=1.5cm]{geometry}
\usepackage[utf8]{inputenc}
\usepackage{amsmath,amssymb,mathtools,amsthm,stmaryrd}
\usepackage{physics}
\usepackage[group-separator={,}]{siunitx}
\usepackage{todonotes}
\usepackage{tikz}

\usepackage{subfig}

\usepackage{qcircuit}
\usepackage{hyperref}
\usepackage{authblk}

\usetikzlibrary{decorations.pathreplacing}

\title{Space-efficient binary optimization for variational computing}
\author[1]{Adam Glos\thanks{aglos@iitis.pl}}
\author[1]{Aleksandra Krawiec}
\author[2,3]{Zolt\'an Zimbor\'as}
\affil[1]{Institute of Theoretical and Applied Informatics, Polish Academy of 
Sciences, ul. Ba{\l}tycka 5, 44-100 Gliwice, Poland}
\affil[2]{Wigner Research Centre for Physics, H-1525,  P.O.Box 49, Budapest, Hungary}
\affil[3]{BME-MTA Lend\"ulet Quantum Information Theory Research Group, Budapest, Hungary}

\date{}

\newcommand{\ii}{\mathrm i}
\newcommand{\ee}{\mathrm e}

\newcommand{\RR}{\mathbb{R}}

\newtheorem{theorem}{Theorem}

\begin{document}

\twocolumn
\maketitle
\begin{abstract} 
In the era of Noisy Intermediate-Scale Quantum (NISQ) computers it is crucial to
design quantum algorithms which do not require many qubits or deep circuits.
Unfortunately, the most well-known quantum algorithms are too demanding to be
run on currently available quantum devices. Moreover, even the state-of-the-art
algorithms developed for the NISQ era often suffer from high space complexity
requirements for particular problem classes.

In this paper, we show that it is possible to greatly reduce the number of
qubits needed for the Travelling Salesman Problem (TSP), a paradigmatic
optimization task, at the cost of having deeper variational circuits. While the focus is on
this particular problem, we claim that the approach can be generalized for other
problems where the standard bit-encoding is highly inefficient. Finally, we also
propose encoding schemes which smoothly interpolate between the qubit-efficient
and the circuit depth-efficient models.
All the proposed encodings remain efficient to implement within the Quantum Approximate Optimization Algorithm framework.


\end{abstract} 


\section{Introduction} During the past few decades, a wealth of quantum
algorithms have been designed for various problems, many of which offered a
speedup over their classical equivalents
\cite{grover1996fast,shor1999polynomial, montanaro2016quantum,
	jordan2011quantum}. The theoretical developments have also been complemented by
progress on the experimental side. Indeed, the demonstration of quantum
supremacy by Google \cite{arute2019quantum} indicates  that in the near future 
useful quantum technology may be available. 
However, current Noisy Intermediate-Scale Quantum (NISQ) devices are too small
and too noisy to implement complicated algorithms like Shor's factorization
algorithm \cite{shor1999polynomial} or  Grover-search based optimizers
\cite{montanaro2015quantum, ambainis2019quantum}. 
This resulted in a new field of quantum computation which focuses on designing
new algorithms requiring significantly less noisy qubits.

Optimization is a problem which seems to be particularly suitable for current
NISQ devices. In particular, the Variational Quantum Eigensolver (VQE) \cite{peruzzo2014variational, mcclean2016theory, moll2018quantum} 
seems to be the state-of-the-art algorithm for solving molecule Hamiltonians. 
Although it can
solve optimization problems defined over discrete spaces, so-called
combinatorial optimization problems, quantum annealing and Quantum Approximate
Optimization Algorithm (QAOA) \cite{farhi2014quantum} are considered to be more suitable.

For all of these algorithms, the original combinatorial optimization problem has
to be transformed into the Ising model. Typically, one starts with a high-level
description, like the Max-Cut problem, where nodes in the graph $G=(V,E)$ have 
to be colored either red or black, so that certain function is minimized. 
Then, one has to
transform it to a pseudo-Boolean polynomial $\sum_{\{u,v\}\in E} (b_u-b_v)^2$.
Each binary variable (bit) $b_v$ denotes the color of the node $v$ in the 
graph. For
example we can choose $b_v=1$ for red color, and $b_i=0$ for black color. For
quantum algorithms it is convenient to change the representation into Ising
model via transformation $b_v \leftarrow (1-Z_v)/2$. Here $Z_v$ is a Pauli operator
acting on a qubit corresponding to the node $v$. By transforming the original
objective function into Hamiltonian, we also change the domain of the problem
into the space of quantum states.

Quantum optimization algorithms differ in the way how they solve the problem.
Variational Quantum Eigensolver (VQE) is a heuristic algorithm in which the
quantum circuit is optimized using classical procedure. More precisely, we are
given an \emph{ansatz} $U(\theta)$ which, after fixing the parameter $\theta$,
produces a state $\ket{\theta} = \prod_{i=1}^k U_i(\theta_i )\ket{0}$. The
vector $\theta$ is optimized using classical optimization procedure like
gradient descent or simultaneous perturbation stochastic approximation (SPSA), 
so that $\ket{\theta}$ will be localized at high
quality answer. Due to its generality, VQE is commonly used for molecule
Hamiltonians, however its usability to the classical optimization problems may be limited.

Quantum annealing theoretically can also be applied for chemistry Hamiltonians,
however current machines restrict their usability to combinatorial optimization
problems. The algorithm turns the ground state of initial Hamiltonian $H_{\rm
	mix} = \sum_i X_i$ into a ground state of objective Hamiltonian $H$ through
adiabatic evolution $g(t) H_{\rm mix} + (1-g(t)) H$. Adiabatic theorem provides
a good premise for high-quality solutions of the problem. Furthermore, while
available D-Wave's annealers have thousands of qubits, the topology restrictions
may limit the size of tractable problems to cases solvable by classical
procedures \cite{choi2008minor,choi2011minor}.

Quantum Approximate Optimization Algorithm (QAOA) is a mixture of the methods 
described
above \cite{farhi2014quantum}. While quantum annealing is a continuous process, QAOA interchangeably
applies both $H_{\rm mix}$ and $H$ for some time. The evolution time is a 
parameter of the evolution, and as it was in the case of VQE, they are adjusted
by external classical procedure. Here the resulting state takes the form
\begin{equation}
\ket{\theta} = \prod_{i=1}^r \exp(-\ii \theta_{\rm mix,i} H_{\rm mix}) \exp(-\ii \theta_{\rm obj,i} H) \ket{+},
\end{equation}
where $r$ is the number of levels. The algorithm can be implemented as long as
both mixing and objective Hamiltonians can be implemented, which in particular
allows for applying it to combinatorial problems. Many studies have been performed to characterize properties of QAOA algorithms, including, both rigorous proofs of computational power and reachability properties \cite{morales2020universality, lloyd2018quantum, hastings2019classical, farhi2020quantum} as well as characterization through heuristics and numerical experiments and extensions of the algorithm \cite{hadfield2019quantum,  akshay2020reachability, zhu2020adaptive, wierichs2020avoiding}.

While the proposed quantum algorithms can theoretically solve arbitrary
combinatorial optimization problems, not all pseudo-Boolean polynomials can be
natively considered for current quantum devices. A general pseudo-Boolean
polynomial takes the form
\begin{equation*}
H(b) = \sum_{I\subseteq \{1,\dots,n\}} \alpha_I \prod_{i\in I} b_i,
\end{equation*}
where $\alpha_I\in \RR$ defines the optimization problem. An order of such
Hamiltonian is a maximum size of $S$ for which $\alpha_S$ is nonzero. Current
D-Wave machines are restricted to polynomials of order 2, hence if one would
like to solve Hamiltonians of higher order, first a quadratization procedure has
to be applied, in general at cost of extra qubits
\cite{boros2012quadratization}. Note that optimization of quadratic polynomials,
Quadratic Unconstrained Binary Optimization (QUBO) is NP-complete, hence it
encapsulates most of the relevant problems.

The objective Hamiltonian for Max-Cut requires $n$ qubits for graph of order
$n$. Hence it can encode $2^n$ solutions, which is equal to the number of all
possible colorings. However, this is not the case in general. For example for
Traveling Salesman Problem (TSP) over $N$ cities, the QUBO requires $N^2$ bits
\cite{lucas2014ising}. However, $N!$ permutations requires only $\lceil
\log_2(N!)\rceil \approx N\log N$ bits. We consider this as a waste of
computational resources. Unfortunately, in general polynomials with optimal
number of qubits have order larger than two, thus we are actually dealing with
higher-order binary optimization, which is currently not possible using D-Wave
machines.

The idea of using higher-order terms is not new. In fact, in the original work
of Farhi et al \cite{farhi2014quantum}, the authors have not restricted the
model to two-local model. Furthermore, Hamiltonian of order 4 was used for variational quantum factoring~\cite{anschuetz2019variational}, while Hamiltonian of order $l$ was constructed for Max-$l$-SAT problem \cite{hadfield2019quantum}. Since the terms of arbitrary
order can be implemented efficiently, QAOA for the problem can reduce the number
of required logical qubits. In general, if objective polynomial is of constant
order $\alpha$, then the circuit of order $\order{\binom{n}{\alpha} \log
	\alpha}$ implements the objective Hamiltonian exactly. While the number may be
large, it is still polynomial, which makes the implementation tractable.
However, even for slowly growing $\alpha$ (say $\order{\log n}$), in general the
number of terms grows exponentially, which could be the case for $l\to\infty$ in
Max-$l$-SAT. Note that even an encoding that requires only logarithmic number of qubits has been introduced \cite{tan2020qubit}, however, the minimizer of this encoding does not necessarily map to the minimizer of the original problem.


Furthermore, when dealing with unbounded order, one has to be careful when
transforming QUBO into the Ising model. Let us consider a polynomial
$2^n\prod_{i=1}^nb_i$. Using default transformation $b_i\leftarrow (1-Z_i)/2$
will produce Hamiltonian $\sum_{I\subseteq \{1,\dots, n\}} (-1)^{|I|} 
\prod_{i\in I} Z_i$, which consists of $2^n$ terms \cite{mandal2020compressed}.
For this example, one can easily find a better Hamiltonian $- \sum_{i=1}^N Z_i$,
that shares the same global minimizer, however, in general finding such a
transformation requires a higher-level understanding of the problem. Note that
this is not an issue for constant-order polynomials, as the number of terms is
guaranteed to be polynomial even in the worst case scenario.

Despite the potential issues coming from utilizing unbounded-order
polynomials, we present a polynomial (encoding) for TSP problem with unbounded
order, which can be efficiently implemented using approximately optimal
number of qubits. Furthermore, our model requires fewer measurements for
estimating energy. We also developed a transition encoding, where one can
adjust the improvement in the required number of qubits and circuit's depth. 
Finally, QAOA optimizes our encoding with similar or  better efficiency 
compared to the state-of-the-art QUBO encoding of TSP.

\paragraph{State-of-the-art TSP encoding } Current state-of-the-art encoding of
TSP problem can be found in the paper by Lucas
\cite{lucas2014ising}. Let us consider the Traveling Salesman Problem (TSP)
over $N$ cities. Let $W$ be a cost matrix, and $b_{ti}$ be a binary variable
such that $b_{ti}=1$ iff the $i$-th city is visited at time $t$. The QUBO
encoding takes the form \cite{lucas2014ising}
\begin{equation*}
\begin{split}
H^{\rm QUBO}(b) &= A_1 \sum_{t=1}^N \left( 1 -\sum_{i=1}^N b_{ti} \right)^{\!\!2} + A_2 \sum_{i=1}^N \left( 1 -\sum_{t=1}^N b_{ti} \right)^{\!\!2} \\
&\phantom{\ =}+ B\sum_{\substack{i,j=1\\i\neq j}}^N W_{ij} \sum_{t=1}^N b_{ti}b_{t+1,j}.
\end{split}
\end{equation*}
Here $A_1,A_2> B\max _{i\neq j}W_{ij}$ are parameters that have to be adjusted
during the optimization. We also assume $N+1\to1$ simplification for the
indices. Note that any route can be represented in $N$ different ways, depending
on which city is visited at time $t=1$. Such redundancy can be solved by fixing
that the first city should be visited at time $t=1$. Thanks to that, $n=(N-1)^2$
(qu)bits in total are required.

In the scope of this paper we will take advantage of various quality 
measures of encodings. First, since the 
Hamiltonian has to be implemented directly, we prefer encodings with possibly
small depth. In this manner, QUBO can be simulated with a circuit of depth 
$\order{n}$ using round-robin scheduling, see Fig.~\ref{fig:schedule}. 

However, QUBO encoding for TSP is inefficient in the number of qubits. Using
Stirling formula one can show that $\lceil \log(N!)\rceil = N \log(N) - N
\log(\ee) + \Theta(\log(N))$ qubits are sufficient to encode all possible
routes, which is significantly smaller than $N^2$. Note that the number of
qubits also has an impact on the volume  of the circuit, traditionally defined
as a product of the number of qubits and the circuit's depth. In case of this
encoding the volume is of order $\order{N^3}$.  

Finally, in the paper we also consider the required number of measurements to
estimate the energy within constant additive error. Instead of estimating each
term of the Hamiltonian independently, which has to be done for VQE, we consider
the measurement's output as a single sample. This way, using Hoeffding's
theorem, QUBO encoding requires $\order{N^3}$ measurements.

\begin{figure*}\centering

	\begin{tikzpicture}
	
	\newcommand{\sca}{0.6}  
	\newcommand{\dist}{0.05} 
	\newcommand{\cityheight}{3.8}  
	
	\newcommand{\blacksqtspqubo}[2]{
		\draw [fill=black] (\sca*#1+\dist,\sca*#2+\dist) rectangle 
		(\sca*#1 + \sca - \dist,\sca*#2 + \sca - \dist);}
	\newcommand{\whitesqtspqubo}[2]{
		\draw [] (\sca*#1+\dist,\sca*#2+\dist) rectangle 
		(\sca*#1 + \sca-\dist,\sca*#2 + \sca-\dist);}
	
	
	\newcommand{\blacksqtsphobol}[2]{
		\draw [fill=black] (\sca*#1 + 3*\dist,\sca*#2 + \dist) rectangle 
		(\sca*#1 + \sca + \dist,\sca*#2 + \sca - \dist);}
	\newcommand{\whitesqtsphobol}[2]{
		\draw [] (\sca*#1 + 3*\dist,\sca*#2 + \dist) rectangle 
		(\sca*#1 + \sca + \dist,\sca*#2 + \sca - \dist);}
	\newcommand{\blacksqtsphobor}[2]{
		\draw [fill=black] (\sca*#1 - \dist,\sca*#2 + \dist) rectangle 
		(\sca*#1 + \sca - 3*\dist,\sca*#2 + \sca - \dist);}
	\newcommand{\whitesqtsphobor}[2]{
		\draw [] (\sca*#1 - \dist,\sca*#2 + \dist) rectangle 
		(\sca*#1 + \sca - 3*\dist,\sca*#2 + \sca - \dist);}
	
	
	\newcommand{\blacksqtspmixedl}[2]{
		\draw [fill=black] (\sca*#1 + 2*\dist,\sca*#2 + \dist) rectangle 
		(\sca*#1 + \sca,\sca*#2 + \sca - \dist);}
	\newcommand{\blacksqtspmixedr}[2]{
		\draw [fill=black] (\sca*#1,\sca*#2 + \dist) rectangle 
		(\sca*#1 + \sca - 2*\dist,\sca*#2 + \sca - \dist);}
	\newcommand{\whitesqtspmixedl}[2]{
		\draw [] (\sca*#1 + 2*\dist,\sca*#2 + \dist) rectangle 
		(\sca*#1 + \sca,\sca*#2 + \sca - \dist);}
	\newcommand{\whitesqtspmixedr}[2]{
		\draw [] (\sca*#1,\sca*#2 + \dist) rectangle 
		(\sca*#1 + \sca - 2*\dist,\sca*#2 + \sca - \dist);}

	\node at (0,0) {
	\begin{tikzpicture}[transform shape,line width=0.2pt]
	\foreach \x in {1,...,6}{%
		\pgfmathparse{-(\x)*60+120}
		\node[draw,circle,inner sep=0.1cm] (N-\x) at 
		(\pgfmathresult:2.cm) [thick] {$C_\x$};
	} 
	\foreach \x [count=\xi from 1] in {2,...,6}{%
		\foreach \y in {\x,...,6}{%
			\path (N-\xi) edge[-] (N-\y);
		}
	}
	
	\draw [ultra thick, ->, color=blue] (N-1) -- (N-3);
	\draw [ultra thick, ->, color=blue] (N-3) -- (N-6);
	\draw [ultra thick, ->, color=blue] (N-6) -- (N-2);
	\draw [ultra thick, ->, color=blue] (N-2) -- (N-4);
	\draw [ultra thick, ->, color=blue] (N-4) -- (N-5);
	\draw [ultra thick, ->, color=blue] (N-5) -- (N-1);
	\end{tikzpicture}};
	
	\node at (6,2) {
	\begin{tikzpicture}[scale=.8]  
	\foreach \x in {0,...,5}
	\foreach \y in {0,...,5} 
	\whitesqtspqubo{\x}{\y};
	
	\blacksqtspqubo{0}{5}
	\blacksqtspqubo{1}{2}
	\blacksqtspqubo{2}{4}
	\blacksqtspqubo{3}{1}
	\blacksqtspqubo{4}{0}
	\blacksqtspqubo{5}{3}
	
	\node (t0) at (-.4, \sca*5 + \sca/2) {$b_1$};
	\node (t1) at (-.4, \sca*4 + \sca/2) {$b_2$};
	\node (t2) at (-.4, \sca*3 + \sca/2) {$b_3$};
	\node (t3) at (-.4, \sca*2 + \sca/2) {$b_4$};
	\node (t4) at (-.4, \sca*1 + \sca/2) {$b_5$};
	\node (t5) at (-.4, \sca*0 + \sca/2) {$b_6$};

	\node (Time) at (-1.2, 2.8*\sca) {\rotatebox[]{90}{Time}};
	\end{tikzpicture}};
	
	\node at (10.5,2) {\begin{tikzpicture}[scale=.8]  
		
		\foreach \y in {0,...,5} 
		\whitesqtsphobol{0}{\y};
		
		\foreach \y in {0,...,5} 
		\whitesqtspqubo{1}{\y};
		
		\foreach \y in {0,...,5} 
		\whitesqtsphobor{2}{\y};
		
		\blacksqtspqubo{1}{4}
		
		\blacksqtsphobol{0}{3}
		\blacksqtsphobor{2}{3}
		
		\blacksqtsphobor{2}{2}
		
		\blacksqtspqubo{1}{1}
		\blacksqtsphobor{2}{1}
		
		\blacksqtsphobol{0}{0}
		
		\node (t0) at (-.4, \sca*5 + \sca/2) {$b_1$};
		\node (t1) at (-.4, \sca*4 + \sca/2) {$b_2$};
		\node (t2) at (-.4, \sca*3 + \sca/2) {$b_3$};
		\node (t3) at (-.4, \sca*2 + \sca/2) {$b_4$};
		\node (t4) at (-.4, \sca*1 + \sca/2) {$b_5$};
		\node (t5) at (-.4, \sca*0 + \sca/2) {$b_6$};
		
		\node (Time) at (-1.2, 2.8*\sca) {\rotatebox[]{90}{Time}};
		\end{tikzpicture}};
	
	\node at (6, -2) {\begin{tikzpicture}[scale=.8]   
		\foreach \x in {0,...,5}
		\foreach \y in {4,...,5} 
		\whitesqtspqubo{\x}{\y};

		\foreach \y in {0,...,3} 
		\whitesqtsphobol{0}{\y};
		
		\foreach \y in {0,...,3} 
		\whitesqtspqubo{1}{\y};
		
		\foreach \y in {0,...,3} 
		\whitesqtsphobor{2}{\y};

		\blacksqtspqubo{0}{5}
		\blacksqtspqubo{2}{4}

		\blacksqtsphobol{0}{3}
		\blacksqtsphobor{2}{3}
		
		\blacksqtsphobor{2}{2}
		
		\blacksqtspqubo{1}{1}
		\blacksqtsphobor{2}{1}
		
		\blacksqtsphobol{0}{0}
		
		\node (t0) at (-.4, \sca*5 + \sca/2) {$b_1$};
		\node (t1) at (-.4, \sca*4 + \sca/2) {$b_2$};
		\node (t2) at (-.4, \sca*3 + \sca/2) {$b_3$};
		\node (t3) at (-.4, \sca*2 + \sca/2) {$b_4$};
		\node (t4) at (-.4, \sca*1 + \sca/2) {$b_5$};
		\node (t5) at (-.4, \sca*0 + \sca/2) {$b_6$};
		
		\node (Time) at (-1.2, 2.8*\sca) {\rotatebox[]{90}{Time}};
		\end{tikzpicture}};
	
	\node at (10.5,-2) {\begin{tikzpicture}[scale=.8]   
		
		\foreach \y in {0,...,5} 
		\whitesqtspmixedl{0}{\y};
		
		\foreach \y in {0,...,5} 
		\whitesqtspmixedr{1}{\y};
		
		\foreach \y in {0,...,5} 
		\whitesqtspmixedl{2}{\y};
		
		\foreach \y in {0,...,5} 
		\whitesqtspmixedr{3}{\y};
		
				\blacksqtspmixedr{1}{5}
		
		\blacksqtspmixedr{1}{4}
		\blacksqtspmixedl{0}{4}
		
		\blacksqtspmixedl{2}{3}
		\blacksqtspmixedr{3}{3}
		
		\blacksqtspmixedl{0}{2}
		
		\blacksqtspmixedr{3}{1}
		
		\blacksqtspmixedl{2}{0}
		
		\node (t0) at (-.4, \sca*5 + \sca/2) {$b_1$};
		\node (t1) at (-.4, \sca*4 + \sca/2) {$b_2$};
		\node (t2) at (-.4, \sca*3 + \sca/2) {$b_3$};
		\node (t3) at (-.4, \sca*2 + \sca/2) {$b_4$};
		\node (t4) at (-.4, \sca*1 + \sca/2) {$b_5$};
		\node (t5) at (-.4, \sca*0 + \sca/2) {$b_6$};
		
		\node (Time) at (-1.2, 2.8*\sca) {\rotatebox[]{90}{Time}};
		\end{tikzpicture}};
	
	\node at (-2.5,2) {\bf \large a)};
	\node at (4.,3) {\bf \large b)};
	\node at (9.3,3) {\bf \large c)};
	\node at (4,-1) {\bf \large d)};
	\node at (9.,-1) {\bf \large e)};
	\end{tikzpicture}

	\caption{\label{fig:encodings} Visualization of QUBO encoding and encodings introduced in the paper for TSP problem. a) exemplary solution for TSP problem. On the right there are assignments of the exemplary solution using  respectively b) QUBO, c) HOBO d), na\"ive mixed, and e) mixed encodings.}
\end{figure*}
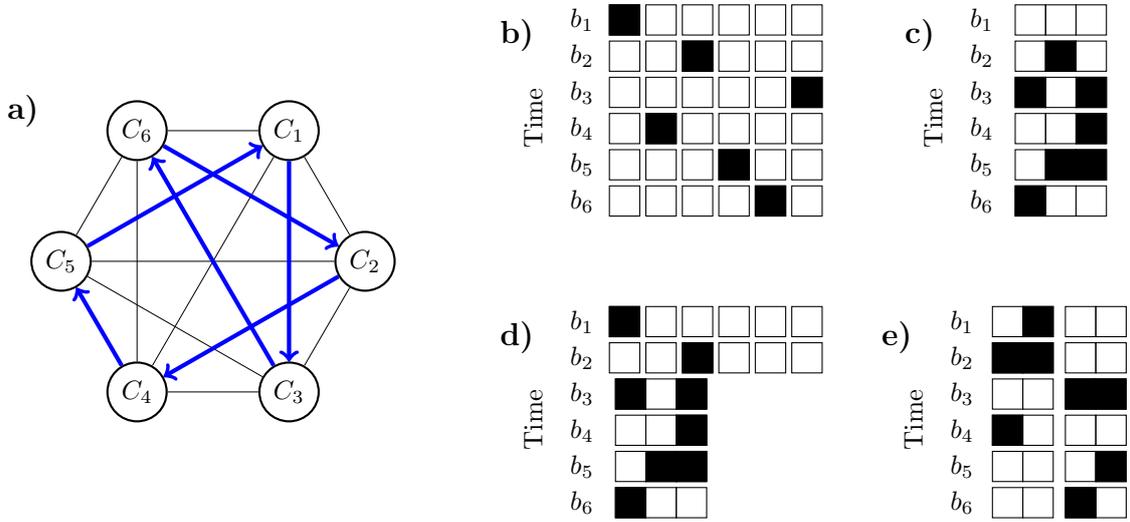

\begin{figure*}\centering 
\begin{tikzpicture}
\newcommand{\len}{3.5}
\newcommand{\hh}{0.17}
\newcommand{\length}{3.5}

\newcommand{\leng}{0.7} 
\newcommand{\height}{0.5} 
\newcommand{\eps}{0.06}

\newcommand{\mygate}[3]{  
\draw [fill = #3] (#1, #2-3*\hh/2) 
	rectangle (#1 + \leng, #2 + 3*\hh/2);
}

\newcommand{\vertline}[3]{ 
\draw [thick] (#1+\leng/2,#2-3*\hh/2) -- (#1+\leng/2,#3 + 3*\hh/2); 
}

\definecolor{color1}{rgb}{0.9, 0.75, 0.54}
\definecolor{color2}{rgb}{0.61, 0.87, 1.0}
\definecolor{color3}{rgb}{0.56, 0.74, 0.56}
\definecolor{color4}{rgb}{0.8, 0.6, 0.8}
\definecolor{color5}{rgb}{0.93, 0.91, 0.67}

\node at (0,2) {
\begin{tikzpicture}[transform shape,line width=0.2pt]

\pgfmathparse{-1*72+162}
	\node[draw,circle,inner sep=0.1cm, fill = color1] (N-1) at 
	(\pgfmathresult:1.1cm) [thick] {$\bf{b_1}$};
\pgfmathparse{-2*72+162}
	\node[draw,circle,inner sep=0.1cm, fill = color2] (N-2) at 
	(\pgfmathresult:1.1cm) [thick] {$\bf{b_2}$};
\pgfmathparse{-3*72+162}
	\node[draw,circle,inner sep=0.1cm, fill = color3] (N-3) at 
	(\pgfmathresult:1.1cm) [thick] {$\bf{b_3}$};
\pgfmathparse{-4*72+162}
	\node[draw,circle,inner sep=0.1cm, fill = color4] (N-4) at 
	(\pgfmathresult:1.1cm) [thick] {$\bf{b_4}$};
\pgfmathparse{-5*72+162}
	\node[draw,circle,inner sep=0.1cm, fill = color5] (N-5) at 
	(\pgfmathresult:1.1cm) [thick] {$\bf{b_5}$};	

\draw [ultra thick, <->] (N-1) -- (N-4);
\draw [ultra thick, <->] (N-2) -- (N-3);
\end{tikzpicture}
};

\node at (\len,2) {
\begin{tikzpicture}[transform shape,line width=0.2pt]

\pgfmathparse{(-1)*72+162-1*72}
	\node[draw,circle,inner sep=0.1cm, fill = color1] (N-1) at 
	(\pgfmathresult:1.1cm) [thick] {$\bf{b_1}$};
\pgfmathparse{(-2)*72+162-1*72}
	\node[draw,circle,inner sep=0.1cm, fill = color2] (N-2) at 
	(\pgfmathresult:1.1cm) [thick] {$\bf{b_2}$};
\pgfmathparse{(-3)*72+162-1*72}
	\node[draw,circle,inner sep=0.1cm, fill = color3] (N-3) at 
	(\pgfmathresult:1.1cm) [thick] {$\bf{b_3}$};
\pgfmathparse{(-4)*72+162-1*72}
	\node[draw,circle,inner sep=0.1cm, fill = color4] (N-4) at 
	(\pgfmathresult:1.1cm) [thick] {$\bf{b_4}$};
\pgfmathparse{(-5)*72+162-1*72}
	\node[draw,circle,inner sep=0.1cm, fill = color5] (N-5) at 
	(\pgfmathresult:1.1cm) [thick] {$\bf{b_5}$};	

\draw [ultra thick, <->] (N-1) -- (N-2);
\draw [ultra thick, <->] (N-5) -- (N-3);
\end{tikzpicture}
};

\node at (2*\len,2) {
\begin{tikzpicture}[transform shape,line width=0.2pt]

\pgfmathparse{(-1)*72+162-2*72}
	\node[draw,circle,inner sep=0.1cm, fill = color1] (N-1) at 
	(\pgfmathresult:1.1cm) [thick] {$\bf{b_1}$};
\pgfmathparse{(-2)*72+162-2*72}
	\node[draw,circle,inner sep=0.1cm, fill = color2] (N-2) at 
	(\pgfmathresult:1.1cm) [thick] {$\bf{b_2}$};
\pgfmathparse{(-3)*72+162-2*72}
	\node[draw,circle,inner sep=0.1cm, fill = color3] (N-3) at 
	(\pgfmathresult:1.1cm) [thick] {$\bf{b_3}$};
\pgfmathparse{(-4)*72+162-2*72}
	\node[draw,circle,inner sep=0.1cm, fill = color4] (N-4) at 
	(\pgfmathresult:1.1cm) [thick] {$\bf{b_4}$};
\pgfmathparse{(-5)*72+162-2*72}
	\node[draw,circle,inner sep=0.1cm, fill = color5] (N-5) at 
	(\pgfmathresult:1.1cm) [thick] {$\bf{b_5}$};	

\draw [ultra thick, <->] (N-4) -- (N-2);
\draw [ultra thick, <->] (N-5) -- (N-1);
\end{tikzpicture}
};

\node at (3*\len,2) {
\begin{tikzpicture}[transform shape,line width=0.2pt]
B
\pgfmathparse{(-1)*72+162-3*72}
	\node[draw,circle,inner sep=0.1cm, fill = color1] (N-1) at 
	(\pgfmathresult:1.1cm) [thick] {$\bf{b_1}$};
\pgfmathparse{(-2)*72+162-3*72}
	\node[draw,circle,inner sep=0.1cm, fill = color2] (N-2) at 
	(\pgfmathresult:1.1cm) [thick] {$\bf{b_2}$};
\pgfmathparse{(-3)*72+162-3*72}
	\node[draw,circle,inner sep=0.1cm, fill = color3] (N-3) at 
	(\pgfmathresult:1.1cm) [thick] {$\bf{b_3}$};
\pgfmathparse{(-4)*72+162-3*72}
	\node[draw,circle,inner sep=0.1cm, fill = color4] (N-4) at 
	(\pgfmathresult:1.1cm) [thick] {$\bf{b_4}$};
\pgfmathparse{(-5)*72+162-3*72}
	\node[draw,circle,inner sep=0.1cm, fill = color5] (N-5) at 
	(\pgfmathresult:1.1cm) [thick] {$\bf{b_5}$};	

\draw [ultra thick, <->] (N-3) -- (N-1);
\draw [ultra thick, <->] (N-4) -- (N-5);
\end{tikzpicture}
};

\node at (4*\len,2) {
\begin{tikzpicture}[transform shape,line width=0.2pt]

\pgfmathparse{(-1)*72+162-4*72}
	\node[draw,circle,inner sep=0.1cm, fill = color1] (N-1) at 
	(\pgfmathresult:1.1cm) [thick] {$\bf{b_1}$};
\pgfmathparse{(-2)*72+162-4*72}
	\node[draw,circle,inner sep=0.1cm, fill = color2] (N-2) at 
	(\pgfmathresult:1.1cm) [thick] {$\bf{b_2}$};
\pgfmathparse{(-3)*72+162-4*72}
	\node[draw,circle,inner sep=0.1cm, fill = color3] (N-3) at 
	(\pgfmathresult:1.1cm) [thick] {$\bf{b_3}$};
\pgfmathparse{(-4)*72+162-4*72}
	\node[draw,circle,inner sep=0.1cm, fill = color4] (N-4) at 
	(\pgfmathresult:1.1cm) [thick] {$\bf{b_4}$};
\pgfmathparse{(-5)*72+162-4*72}
	\node[draw,circle,inner sep=0.1cm, fill = color5] (N-5) at 
	(\pgfmathresult:1.1cm) [thick] {$\bf{b_5}$};	
	
\draw [ultra thick, <->] (N-2) -- (N-5);
\draw [ultra thick, <->] (N-3) -- (N-4);
\end{tikzpicture}
};

\draw (-0.35*\len,0) -- (4.4*\len, 0);
\draw (-0.35*\len,-\hh) -- (4.4*\len, -\hh);
\draw (-0.35*\len,-2*\hh) -- (4.4*\len, -2*\hh);

\draw (-0.35*\len,-4*\hh) -- (4.4*\len, -4*\hh);
\draw (-0.35*\len,-5*\hh) -- (4.4*\len, -5*\hh);
\draw (-0.35*\len,-6*\hh) -- (4.4*\len, -6*\hh);

\draw (-0.35*\len,-8*\hh) -- (4.4*\len, -8*\hh);
\draw (-0.35*\len,-9*\hh) -- (4.4*\len, -9*\hh);
\draw (-0.35*\len,-10*\hh) -- (4.4*\len, -10*\hh);

\draw (-0.35*\len,-12*\hh) -- (4.4*\len, -12*\hh);
\draw (-0.35*\len,-13*\hh) -- (4.4*\len, -13*\hh);
\draw (-0.35*\len,-14*\hh) -- (4.4*\len, -14*\hh);

\draw (-0.35*\len,-16*\hh) -- (4.4*\len, -16*\hh);
\draw (-0.35*\len,-17*\hh) -- (4.4*\len, -17*\hh);
\draw (-0.35*\len,-18*\hh) -- (4.4*\len, -18*\hh);

\mygate{-0.9}{-1*\hh}{color1};
\mygate{-0.9}{-13*\hh}{color4};
\mygate{0.3}{-5*\hh}{color2};
\mygate{0.3}{-9*\hh}{color3};
\vertline{-0.9}{-1*\hh}{-13*\hh}
\vertline{0.3}{-5*\hh}{-9*\hh}

\mygate{3.2}{-9*\hh}{color3};
\mygate{3.2}{-17*\hh}{color5};
\mygate{3.2}{-1*\hh}{color1};
\mygate{3.2}{-5*\hh}{color2};
\vertline{3.2}{-9*\hh}{-17*\hh}
\vertline{3.2}{-1*\hh}{-5*\hh}

\mygate{6.1}{-5*\hh}{color2};
\mygate{6.1}{-13*\hh}{color4};
\mygate{7.3}{-1*\hh}{color1};
\mygate{7.3}{-17*\hh}{color5};
\vertline{6.1}{-5*\hh}{-13*\hh}
\vertline{7.3}{-1*\hh}{-17*\hh}

\mygate{10.15}{-1*\hh}{color1};
\mygate{10.15}{-9*\hh}{color3};
\mygate{10.15}{-13*\hh}{color4};
\mygate{10.15}{-17*\hh}{color5};
\vertline{10.15}{-1*\hh}{-9*\hh}
\vertline{10.15}{-13*\hh}{-17*\hh}

\mygate{13.1}{-5*\hh}{color2};
\mygate{13.1}{-17*\hh}{color5};
\mygate{14.3}{-9*\hh}{color3};
\mygate{14.3}{-13*\hh}{color4};
\vertline{13.1}{-5*\hh}{-17*\hh}
\vertline{14.3}{-9*\hh}{-13*\hh}

\draw [decorate,decoration={brace,amplitude=3pt},xshift=2pt,yshift=0pt]
(-0.4*\len,-2*\hh-\eps) -- (-0.4*\len,0*\hh+\eps) node 
[black,midway,xshift=-0.35cm] {{\large$\bf{b_1}$}};
\draw [decorate,decoration={brace,amplitude=3pt},xshift=2pt,yshift=0pt]
(-0.4*\len,-6*\hh-\eps) -- (-0.4*\len,-4*\hh+\eps) node 
[black,midway,xshift=-0.35cm] {{\large$\bf{b_2}$}};
\draw [decorate,decoration={brace,amplitude=3pt},xshift=2pt,yshift=0pt]
(-0.4*\len,-10*\hh-\eps) -- (-0.4*\len,-8*\hh+\eps) node 
[black,midway,xshift=-0.35cm] {{\large$\bf{b_3}$}};
\draw [decorate,decoration={brace,amplitude=3pt},xshift=2pt,yshift=0pt]
(-0.4*\len,-14*\hh-\eps) -- (-0.4*\len,-12*\hh+\eps) node 
[black,midway,xshift=-0.35cm] {{\large$\bf{b_4}$}};
\draw [decorate,decoration={brace,amplitude=3pt},xshift=2pt,yshift=0pt]
(-0.4*\len,-18*\hh-\eps) -- (-0.4*\len,-16*\hh+\eps) node 
[black,midway,xshift=-0.35cm] {{\large$\bf{b_5}$}};
\end{tikzpicture}
\caption{\label{fig:schedule} Round-robin schedule \cite{rasmussen2008round} 
for binary vectors $b_1,\dots,b_5$. We assumed that each $b_i$ is defined over 
3 qubits. Each gate defined over a pair $b_i,b_j$ is an implementation of 
the Hamiltonian defined over these variables. Such Hamiltonian can be 
implemented using the technique visualized in 
Fig.~\ref{fig:decomposition-and-gray-code}.}
\end{figure*}
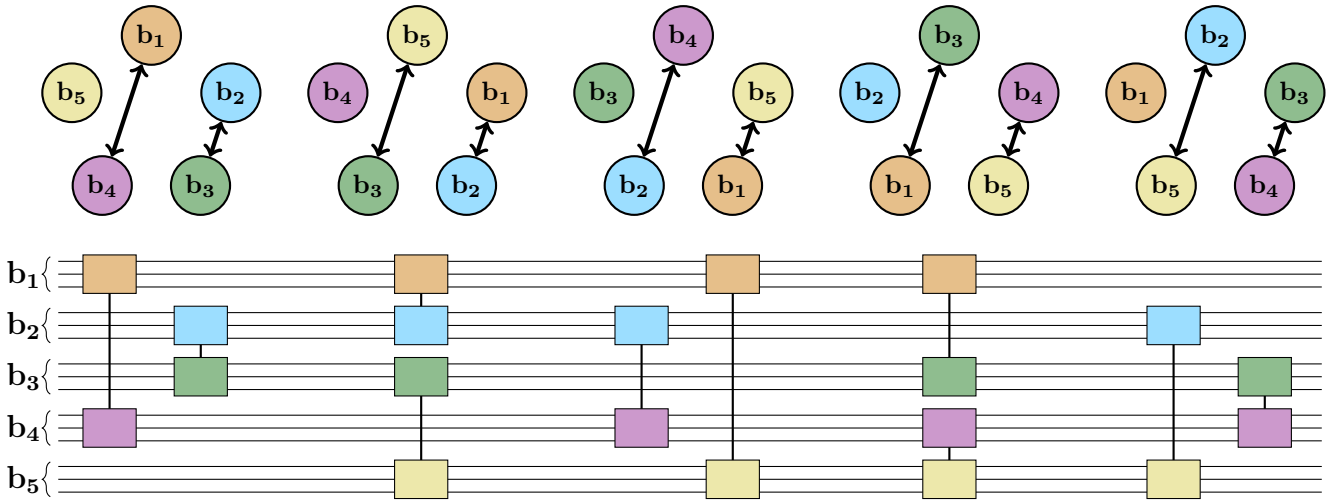

\section*{Results}

\paragraph{Preliminaries}

Traveling Salesman Problem is natively defined  over the permutations of
$\{1,\dots,N\}$. A simple encoding can be defined as follows. We make a 
partition of all bits into $N$ collections $b_t$, where each collection  
encodes a city visited in a particular time. Then, for each collection we 
choose a number encoding which represents the city.

QUBO is an example of such an encoding, where each city is represented by an 
one-hot vector, see Fig.~\ref{fig:encodings}. Instead, each city can be encoded 
as a number using binary numbering system. 
Using numbering system is a state-of-the-art way 
to encode inequalities \cite{lucas2014ising}, however it is new in the context 
of encoding elements of a feasible space. 

The Hamiltonian takes the form
\begin{equation}\label{eq:general-tsp-hobo}
\begin{split}
H(b) &= A_1\sum_{t=1}^N H_{\rm valid} (b_t) + A_2\sum_{t=1}^N \sum_{t'=t+1}^N H_{\neq}(b_t,b_{t'}) \\
&\phantom{\ =}+ B\sum_{\substack{i,j=1\\i\neq j}}^N W_{ij} \sum_{t=1}^N 
H_\delta (b_t,i)H_\delta (b_{t+1},j).
\end{split}
\end{equation}
Hamiltonian $H_{\rm valid}$ checks whether a vector of bits $b_t$ encodes a 
valid city. For example for QUBO it checks whether at most one bit is equal to 
1. 

Hamiltonian $H_{\neq}$ verifies whether two collections encode the same city. 
Note that QUBO encoding falls into this representation by choosing 
\begin{equation*}
H_{\neq}^{\rm QUBO}(b_t,b_{t'}) = \sum_{i\neq j} \left (2b_{ti}b_{t'j} +\frac{1}{\binom{N}{2}}(1-b_{ti} -b_{t'j})\right ).
\end{equation*}

Hamiltonian $H_\delta$ plays a similar role as $H_{\neq}$. If the inputs are 
different, then both $H_\delta$ and $H_{\neq}$ give zeros. If the inputs are 
the same, then the outputs are nonzero and moreover we expect that the 
Hamiltonian $H_{\neq}$ outputs $1$. 
This is in order to 
preserve the costs of routes. In case of QUBO we took 
$H_\delta(b_t,i) = b_{ti}$. Note that in particular $H_{\delta}=H_{\neq}$ may 
be a good choice, 
however later we will show that choosing different $H_{\neq}$ may be beneficial.

\begin{figure*}
\begin{tikzpicture}
\newcommand{\length}{0.75} 
\newcommand{\height}{0.6} 

\newcommand{\len}{1.5} 
\newcommand{\hh}{0.7} 

\newcommand{\lensec}{0.8} 

\newcommand{\step}{0.15} 
\newcommand{\eps}{0.1} 

\newcommand{\singlegate}[3]{
\draw [fill = blue!20] (#1,#2-\height/2) rectangle (#1 + \length,#2 + 
\height/2);
\node (t0) at (#1+\length/2, #2) {$#3$};
}

\newcommand{\doublegate}[4]{  
\draw [fill=blue!20] (#1,#2-\height/2) rectangle (#1 + \length,#3 + \height/2);
\node at (#1+\length/2, #2/2+#3/2) {$#4$};
}

\newcommand{\triplegate}[4]{  
\draw [fill = blue!20] (#1,#2-\height/2) rectangle (#1 + \length,#3 + 
\height/2);
\node at (#1+\length/2, #2/2 + #3/2) {$#4$};
}

\newcommand{\contr}[3]{  
\draw (#1,#2-0.12) -- (#1,#3);
\draw (#1,#2) circle (0.12);
\draw[fill] (#1,#3) circle (0.075);
}

\newcommand{\vertline}[3]{ 
\draw (#1+\length/2,#2+\height/2) -- (#1+\length/2,#3 + \height/2); 
}


\draw (-1*\step,7*\hh) -- (66*\step, 7*\hh);
\draw (-1*\step,6*\hh) -- (66*\step, 6*\hh);
\draw (-1*\step,5*\hh) -- (66*\step, 5*\hh);

\singlegate{3*\step}{5*\hh}{\alpha_3};
\doublegate{12*\step}{5*\hh}{6*\hh}{\alpha_{23}};
\singlegate{21*\step}{6*\hh}{\alpha_2};
\doublegate{30*\step}{6*\hh}{7*\hh}{\alpha_{12}};
\triplegate{39*\step}{5*\hh}{7*\hh}{\alpha_{123}};
\singlegate{48*\step}{5*\hh}{\alpha_{13}};
\vertline{48*\step}{5*\hh}{7*\hh};
\singlegate{48*\step}{7*\hh}{\alpha_{13}};
\singlegate{57*\step}{7*\hh}{\alpha_1};

\draw (-1*\step,0) -- (66*\step, 0);
\draw (-1*\step,\hh) -- (66*\step, \hh);
\draw (-1*\step,2*\hh) -- (66*\step, 2*\hh);
\draw (-1*\step,3*\hh) -- (66*\step, 3*\hh);

\contr{1*\step}{0*\hh}{1*\hh};
\singlegate{3*\step}{0}{\alpha_{3}};
\contr{10*\step}{0*\hh}{2*\hh};
\singlegate{12*\step}{0}{\alpha_{23}};
\contr{19*\step}{0*\hh}{1*\hh};
\singlegate{21*\step}{0}{\alpha_{2}};
\contr{28*\step}{0*\hh}{3*\hh};
\singlegate{30*\step}{0}{\alpha_{12}};

\contr{37*\step}{0*\hh}{1*\hh};
\singlegate{39*\step}{0}{\alpha_{123}};
\contr{46*\step}{0*\hh}{2*\hh};
\singlegate{48*\step}{0}{\alpha_{13}};
\contr{55*\step}{0*\hh}{1*\hh};
\singlegate{57*\step}{0}{\alpha_{1}};
\contr{64*\step}{0*\hh}{3*\hh};

\node at (-2.5*\step,0) {$\ket{0}$};
\node at (67.5*\step,0) {$\ket{0}$};

\draw [decorate,decoration={brace,amplitude=3pt},xshift=0pt,yshift=0pt]
(-2*\step,\hh-\eps) -- (-2*\step,3*\hh+\eps) node 
[black,midway,xshift=-0.4cm] {$\ket\varphi$};
\draw [decorate,decoration={brace,amplitude=3pt},xshift=0pt,yshift=0pt]
(-2*\step,5*\hh-\eps) -- (-2*\step,7*\hh+\eps) node 
[black,midway,xshift=-0.4cm] {$\ket\varphi$};


\node at (-12.5*\step,0) {$\ket{0}$};
\node at (-42*\step,0) {$\ket{0}$};

\newcommand{\hhl}{0.78} 

\draw (-40.5*\step,0) -- (-14.5*\step, 0);
\draw (-40.5*\step,1*\hh) -- (-14.5*\step, 1*\hh);
\draw (-40.5*\step,2*\hh) -- (-14.5*\step, 2*\hh);
\draw (-40.5*\step,3*\hh) -- (-14.5*\step, 3*\hh);

\contr{-17*\step}{0*\hh}{3*\hh};
\contr{-20*\step}{0*\hh}{2*\hh};
\contr{-23*\step}{0*\hh}{1*\hh};

\singlegate{-30*\step}{0}{\alpha};

\contr{-32*\step}{0*\hh}{1*\hh};
\contr{-35*\step}{0*\hh}{2*\hh};
\contr{-38*\step}{0*\hh}{3*\hh};

\draw [decorate,decoration={brace,amplitude=3pt},xshift=0pt,yshift=0pt]
(-41.5*\step,\hh-\eps) -- (-41.5*\step,3*\hh+\eps) node 
[black,midway,xshift=-0.4cm] {$\ket\psi$};

\newcommand{\hhltth}{7*\hh} 
\newcommand{\hhlttw}{6*\hh} 
\newcommand{\hhlto}{5*\hh} 

\draw (-40.5*\step,\hhlto) -- (-14.5*\step, \hhlto);
\draw (-40.5*\step,\hhlttw) -- (-14.5*\step, \hhlttw);
\draw (-40.5*\step,\hhltth) -- (-14.5*\step, \hhltth);

\contr{-17*\step}{\hhlttw}{\hhltth};
\contr{-21.5*\step}{\hhlto}{\hhlttw};

\singlegate{-30*\step}{\hhlto}{\alpha};

\contr{-33.5*\step}{\hhlto}{\hhlttw};
\contr{-38*\step}{\hhlttw}{\hhltth};

\draw [decorate,decoration={brace,amplitude=3pt},xshift=0pt,yshift=0pt]
(-41.5*\step,\hhlto-\eps) -- (-41.5*\step,\hhltth+\eps) node 
[black,midway,xshift=-0.4cm] {$\ket\psi$};

\node at (-6.5,5.5) {\textbf{(a)}};
\node at (-6.5,2.5) {\textbf{(b)}};
\node at (-0.5,5.5) {\textbf{(c)}};
\node at (-0.5,2.5) {\textbf{(d)}};

\end{tikzpicture}

\caption{\label{fig:decomposition-and-gray-code} On the left, two 
decompositions of $\exp(\ii t Z_1 Z_2 Z_3)$, a) without auxiliary qubit 
\cite{seeley2012bravyi} (p.30), and b) with auxiliary qubit 
\cite{nielsen2001quantum} (p.210). On the right, an example of simplifying 
circuit for $\sum_{I\subseteq \{1,2,3\}}\alpha_I \prod_{i\in I} b_i$ using Gray 
codes, c) using $Z$ operations only, and d) with CNOTs and single-qubit gates 
only. In all the figures blue gates are a $k$-local $Z$ operations.}
\end{figure*}
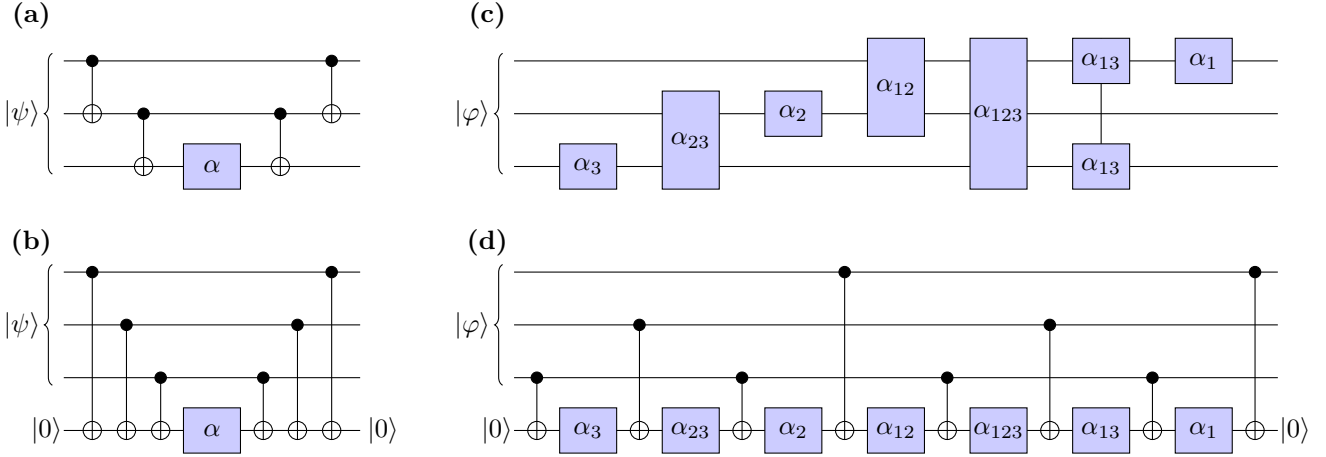

\paragraph{Simple HOBO encoding}
The simplest encoding is the one in which each collection $b_t$ encodes a  
city in a binary system, see Fig.~\ref{fig:encodings}. In this case, for each time we need 
$K\coloneqq \lceil \log N\rceil $ qubits. In total we need 
$\sim N\log(N)$ qubits, which match the lower bound. Moreover, we have to 
design $H_{\rm valid}$ in such a way that $b_t$ represents the number at 
most $N-1$.

Following Eq.~\eqref{eq:general-tsp-hobo}, it is easy to note that HOBO defined in a way described above, is of polynomial size. 
Note that the sum of Hamiltonians $H_{\neq}$ 
produces at most $\binom{N}{2}2^{2K}$ elements. Furthermore, the terms 
introduced by $H_{\rm valid}^{\rm HOBO}$ and $H_{\delta}^{\rm HOBO}$ are already present in $H_{\neq}^{\rm HOBO}$. 
Hence in total we have $\order{N^4}$ terms, which implies the polynomial size 
and  depth, and thus volume.  

Let us now present an exemplary encoding. Suppose $\tilde b_{K-1}\dots \tilde 
b_{0}$ is a binary representation of $N-1$. Suppose $k^0\in K_0$ are such 
indices that $\tilde b_{k^0} = 0$. Then one can show that 
\begin{equation*}
H_{\rm valid}^{\rm HOBO}(b_t) \coloneqq \sum_{k^0\in K_0} b_{t,k_0}\prod_{k=k^0+1}^{K-1} (1-(b_{t,k} - \tilde b_{k})^2)
\end{equation*}
validates whether the encoding number is at most $N-1$. A detailed proof can be 
found in Supplementary materials, here let us consider an example. Suppose 
$N-1=100101_2$. All the numbers larger than $N-1$ are of the form $11????_2$, 
$101???_2$ or $10011?_2$, where `$?$' denotes as arbitrary bit value. The 
polynomial
\begin{equation}
b_{t5}b_{t4} \  + \  b_{t5}(1-b_{t4})b_{t3} \ +\  
b_{t5}(1-b_{t4})(1-b_{t3})b_{t2}b_{t1}.
\end{equation}
punishes all these forms. At the same time, one can verify that numbers smaller 
than $N-1$ are not punished by the Hamiltonian.

Here, we will consider $H_{\neq}^{\rm HOBO} \equiv H_\delta^{\rm HOBO}$, hence 
it is enough to define the latter only. Hamiltonian $H_{\delta}^{\rm HOBO}$ can 
be defined as
\begin{equation*}
H_\delta^{\rm HOBO}(b,b') \coloneqq \prod_{k=1}^K (1- (b_k - b'_k)^2).
\end{equation*}
Note that if $b'$ is a fixed number like it is in the case of objective 
function implementation in Eq.~\eqref{eq:general-tsp-hobo}, then we simply take 
consecutive bits from binary representation.

Let us estimate the cost of this encoding. As it was previously stated, the 
number of factors is of order $\order{N^4}$. Using round-robin techniques and 
Gray code, see Fig.~\ref{fig:schedule} and 
\ref{fig:decomposition-and-gray-code} one can show 
that the depth of the circuit is $\order{N^3}$, which gives us the volume 
$\order{N^4 \log(N)}$. Note that the Gray code requires additional 
$\lfloor NK/2\rfloor$ qubits, which does not change the final result. Finally, 
in order to achieve a similar quality of energy 
measurements, we need $\order{N^2}$ measurements.

One can expect that higher-order binary optimization may lead to difficult 
landscapes, harder to investigate for optimization algorithm. We have 
investigated TSP encodings with $W\equiv 0$ and random $W$ matrices. The 
results are presented in Fig.~\ref{fig:numerics}. One can see that with the 
same number of Hamiltonians applied, the results are either similar or in favor 
for higher-order encodings. 

\begin{figure*}\centering
	\includegraphics[]{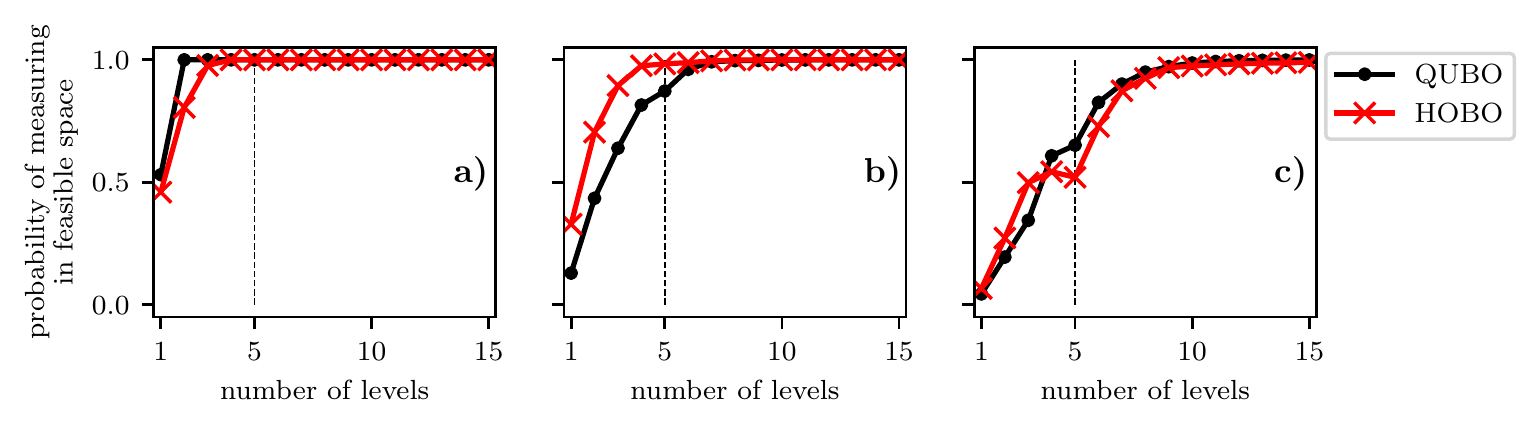}
	\includegraphics[]{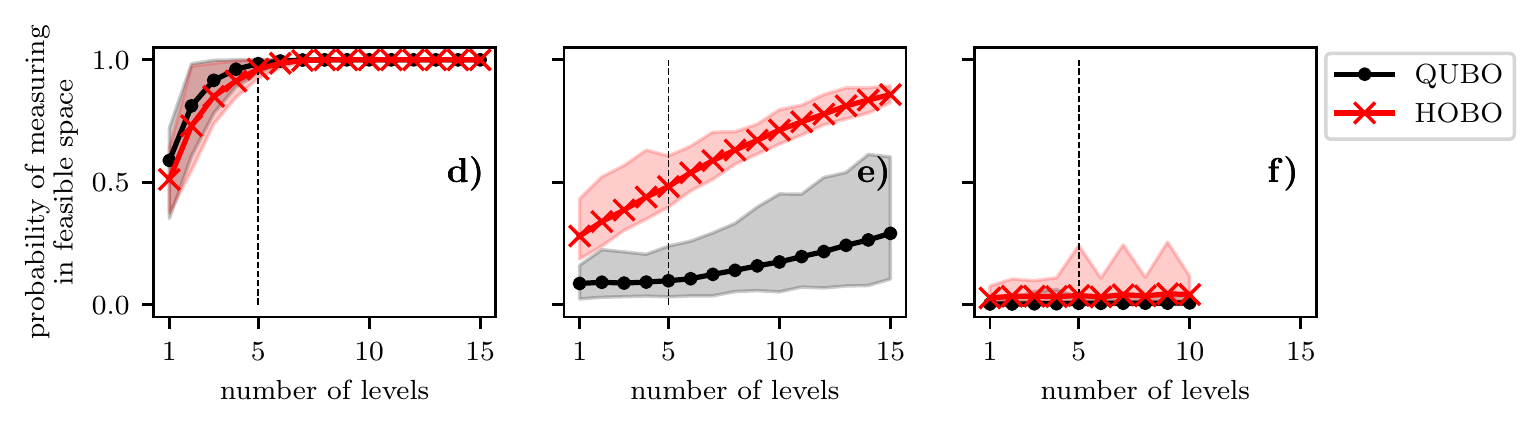}
	
	\caption{\label{fig:numerics} The dependence between the probability of 
	measuring the state in feasible space for a) 3 cities, b) 4 cities, c) 5 
	cities for $W\equiv 0$, and d) 3 cities, e) 4 cities, f) 5 cities for 
	randomly chosen $W$. For most cases HOBO and QUBO present similar quality, 
	while for e) HOBO clearly outperforms QUBO approach. Vertical line denotes 
	the change of optimization method. For f) we were only capable of 
	estimating up to 10 oracles applied due to convergence issues}
\end{figure*}

\paragraph{Mixed QUBO-HOBO approach}

While QUBO encoding requires significantly more qubits compared to HOBO, the 
latter produces much deeper circuits. It is not clear whether the number of 
qubits 
or the depth of the circuit will be more challenging, and in fact we claim that 
both may produce significant difficulties when designing quantum computers.  

One can consider a simple mixing of the proposed QUBO and HOBO approaches in 
the following way: let $R\in\{1,\dots, N\}$ be a free parameter of our model. 
Exactly $R$ of collections $b_t$ will be encoded as one-hot vectors (in QUBO's 
fashion), while the remaining $N-R$ collections will be encoded using the 
binary representation, see Fig. \ref{fig:encodings} d).  

Unfortunately, this approach combines flaws of both models introduced before. 
For $R=\Omega(N)$, the mixed approach requires $\Theta(N^2)$ qubits. On the 
other hand, for $R = \order{N}$ the depth of the circuit is the same as in the 
HOBO approach due to numerous HOBO-encoded $b_t$.

Instead, we propose another encoding. Suppose $N=(2^K-1)L$ for suitable 
integers $K$ and $L$. Each $b_t$ consists of $KL$ qubits of the form $b_{tlk}$. 
The cities are encoded as follows. First $K$ qubits 
(first bunch) decodes numbers $1,\dots,2^K-1$. The second bunch decodes 
$2^K,\dots, 2^K-2$, and so on. Note that QUBO and HOBO encodings introduced 
before are special instances with $L=N$ and $L=1$ respectively.

We add the following assumptions, which also define $H_{\rm valid}$. All bits 
being zero is an invalid assignment, which is equivalent to $\sum_{k,l} 
b_{tlk} \geq 1$. This can be forced by using standard techniques for 
transforming inequalities to QUBO \cite{lucas2014ising}. 
Secondly, if in some bunch there are nonzero bits, then in all other bunches 
bits have to be zeros. Note that this assumption is equivalent to the fact that for all 
$l_0=1,\dots,L$ we have that either $\sum_{k}b_{tl_0k}$ is zero or 
$\sum_{l\neq l_0}\sum_{k}b_{tlk}$ is zero. The Hamiltonian $H_{\rm valid}^{\rm MIX}$ 
takes the form
\begin{equation*}
\begin{split}
H_{\rm valid}^{\rm MIX}(b_t) &\coloneqq \left ( -\sum _{l=1}^L \sum_{k=0}^{K-1} b_{tlk} +1 + \sum_{i=0}^{\lceil\log (KL)\rceil }2^i \xi_{t,i} \right )^2 \\
& \phantom{\ =}+  \sum _{l=1}^L \left ( \sum_{k=0}^{K-1} b_{tlk} \right )\left ( \sum_{\substack{l'=1 \\ l'\neq l}}^L\sum_{k=0}^{K-1} b_{tl'k} \right ).
\end{split}
\end{equation*}
Here $\xi_i$ are additional bits for encoding the first assumption. In total, there will be additional $N \log(KL) \leq N\log(N)$ qubits.

Now let us define $H_{\neq }^{\rm MIX}$ Hamiltonian. Since due to $H_{\rm 
valid}^{\rm MIX}$ 
there exist two indices $l_0,l'_0$ such that $b_{tl_0}$ and $b_{t'l'_0}$ are 
nonzero, we only have to check  for consecutive bunches 
$l=1,\dots,L$ if there exists $l$ such that $b_{t,l}$ are nonzero and 
identical. The Hamiltonian $H_{\neq}^{\rm MIX}$ takes the form
\begin{equation*}
\begin{split}
H_{\neq}^{\rm MIX}(b_t,b_{t'}) &\coloneqq \sum_{l=1}^L \left( \sum_{k=0}^{K-1}(b_{tlk}+b_{t',l,k})\right) \times \\
& \phantom{\ \coloneqq}\times\prod_{k=1}^K (1- (b_{t',l,k} - b'_{t,l,k})^2).  
\end{split}
\end{equation*}
Note that the first factor checks whether the bunches are nonzero, while the 
latter is the Kronecker delta implementation as before.

Finally, let us define $H_{\delta}^{\rm MIX}$. Let $\bar l : \{1,\dots ,N\} 
\mapsto 
\{1,\dots,L\}$ be a function which outputs a bunch index denoting the 
$i$-th city. Then it is enough to apply the Kronecker delta on $\bar l(i)$-th 
bunch. Hence $H_{\delta}^{\rm MIX}$ will take the form
\begin{equation*}
H^{\rm MIX}_\delta(b_t,i) \coloneqq \prod_{k=1}^K (1- (b_{t',\bar l(i),k} - b^i_{k})^2),
\end{equation*}
where $b^i_k$ is $k$-th bit of binary representation of $i$.

Let us now calculate the efficiency of the encoding. We will consider only the 
scenario $k=\alpha \log(N)$ for $\alpha\in (0,1)$. First we note that 
$L=\frac{N}{2^K-1} = \Theta(N^{1-\alpha})$. For the proposed mixed encoding, we 
need $NKL + N\log(KL)= \Theta(N^{2-\alpha}\log(N))$ qubits. Hamiltonian can be 
encoded in a circuits of depth $\order{N^{1+2\alpha}}$. This finally gives us 
the volume $\order{N^{3-\alpha}\log(N)}$. All this parameters show a perfect, 
up to poly-log factors, transition between HOBO and QUBO approaches. Finally, 
to achieve constant error of energy estimation, we require $\order{N^{3-\alpha} 
\log N} \max_{i\neq j} W_{ij}$ runs of the circuit.

\paragraph{Optimal encoding}
So far we assumed that all binary variables are split into collections of 
binary variables, such that each collection defines a particular time point. We 
heavily used this assumption, so that the encoding was particularly simple. 
Theretofore it was implementable on a quantum computer, which is necessary 
for QAOA. Nevertheless, dropping this assumption can save us from even 
more qubits. 

Let $H$ be a diagonal Hamiltonian. Then $\bra{\psi}H\ket{\psi} = \sum_{b\in 
\{0,1\}^n} E_b|\braket{b}{\psi}|^2 $, where $E_i$ is the energy value 
corresponding to the bit string $i$. Hence, the statistics from the measurements are 
sufficient to estimate the energy.

Suppose we are given a general combinatorial optimization problem of function 
$f:\mathcal{A} \to \RR$, where $\mathcal{A}$ is a natural feasible space for 
the problem. In the case of TSP, $\mathcal{A}$ would be a collection of all 
permutations of some fixed order. Let $g:\mathcal{A}\to \mathcal B \subseteq 
\{0,1\}^n$ be a bijection function where $n=\lceil\log(|\mathcal A|) \rceil$.  
Let $\tilde g^{\rm inv}(b)$ be an extension of $g^{-1}$ such that it maps some 
penalty energy $E_{\rm pen} > \min_{a\in\mathcal{A}} f(a)$, i.e. $g^{\rm 
inv}(b)= g^{-1}(b)\delta_{b\in \mathcal B} + E_{\rm pen} \delta_{b\not\in 
\mathcal{B}}$. Then provided that $g^{\rm inv}$ can be efficiently computed, we 
can use it to estimate the expected energy directly from the measurement's statistics. Since converting binary representation into numbers takes negligible time, it is enough to provide a procedure for numbering elements of $\mathcal{A}$

We can incorporate this technique to TSP problem as well. In this case the 
simplest way is to use a factorial numbering system in which $i$-th 
digit starting from least significant can be any number between 0 and $i-1$. In 
general $(d_k\dots d_0)_{!} \equiv \sum_{i=0}^k d_0 \cdot i!$. The opposite 
transformation can be done by computing the modulo by consecutive natural 
numbers. Then such representation can be transformed to permutation via Lehmer 
codes, which starting with the most significant factoradic digit, take 
$(k+1)$-th digit of the sequence $(0,1,\dots,k)$. The used digit is removed and 
the procedure repeats for next digits. The taken digits in given order directly 
encodes a permutation.

Since the procedure described above maps consecutive natural numbers to routes, we require only $\lceil \log(N!)\rceil$ qubits, which is optimal for each $N$.
Since arbitrary pseudo-Boolean function can be transformed to pseudo-Boolean 
polynomial, it is as well the case for $f\circ g^{\rm inv}$. Hence there exists 
a diagonal Hamiltonian representing the same optimization problem. However, in 
general such encoding may require exponential number of terms, which makes it 
intractable for QAOA. Hence for such an approach VQE is at the moment 
the preferable quantum algorithm.  The numbers of required measurements will in 
general depend on the choice of $E_{\rm pen}$, however they can be equal to the 
length of any route, or $N\max_{i\neq j} W_{ij} $. By this we can show that the 
number of measurements is approximately $\order{N}\operatorname{range}(W)$, 
which is significantly smaller than for any encoding described before.

\section*{Discussion} The presented results show that it is possible to
significantly reduce the number of required qubits at the cost of having deeper circuits.
Since both the depth and the number of qubits are challenging tasks, we claim
that it is necessary to provide alternative representation allowing trade-offs
between the different measures. 
Our numerical results hint that the increase of the depth might not be that significant for larger system sizes, as one needs fewer levels in the space-efficient embedded version. Thus, it would be interesting to investigate how many fewer levels one needs in the space-efficient encoding scheme

Note that the approach cannot be applied for general problems. For example,
the state-of-the-art representation of the Max-Cut problem over a graph of order 
$N$ requires exactly $N$ qubits. Since the natural space in general is of 
order $2^N$, it seems unlikely to further reduce the number of qubits. However, 
one can expect similar improvements for other permutation-based problems like
max-$k$-coloring problem.

On top of that, while arbitrary HOBO can be turned into QUBO by automatic
quadratization techniques, it remains an open question whether there are simple
techniques which reduce the number of qubits at the cost of additional Pauli terms.
This is due to the fact that quadratization is well defined: if $H:\mathcal 
{B_X}\to \RR$ is 
a general pseudo-Boolean function, then quadratic pseudo-Boolean function
$H':\mathcal{B_X} \times \mathcal{B_Y}\to \RR$ is its quadratization iff for all
$x\in \mathcal{B_X}$
\begin{equation}
H(x) = \min _{y\in \mathcal{B_Y}} H'(x,y).
\end{equation}
Note that $y$ does not have any meaning in the context of original problem $H$.
However, when removing qubits from binary function, we may not be able to
reproduce the original solution. Thus, such (automatic) procedure requires
context of the problem being solved. 


\begin{table*}\centering 
	\setlength\extrarowheight{3pt}\small
	\begin{tabular}{l@{\qquad}c@{\qquad}c@{\qquad}c@{\qquad}c@{\ }}
		\hline
		& QUBO $H^{\rm QUBO}$ & HOBO $H^{\rm HOBO}$ & mixed $H^{\rm MIX}$ & enumerating $H^{\rm VQE}$ \\\hline 
		No. of qubits & $N^2$ & $N  \log(N)$ & $\frac{\alpha}{C}N^{2-\alpha}\log N$ & $N\lceil \log(N!) \rceil$\\
		No. of terms & $2N^3$ & $\frac{1}{2}N^4$ &  $\frac{C}{2}N^{3+\alpha} $ & exponential\\
		circuit depth & $12N$ & $2N^3$ & $2CN^{1+2\alpha}$ & exponential \\
		circuit volume  & $ 12N^3$  & $2N^4\log N$ & $ 2\alpha N^{3+\alpha }\log N$ & exponential\\
		No. of measurements & $\order{N^3}\max_{i\neq j}W_{ij}$ & $\order{N^2}\max_{i\neq j}W_{ij}$ & $\order{N^{3-\alpha}\log N}\max_{i \neq j}W_{ij}$ & $\order{N}\operatorname{range}(W)$\\\hline
	\end{tabular}
	\caption{Resources required for various Hamiltonian encodings. Only leading 
	terms are written, for more exact bounds see Supplementary materials.  The 
	$\order{\cdot}$ does not depend on the choice of $W$. We assumed $B=1$ and 
	$A_1,A_2\leq =\order{\max_{i\neq j}W_{ij}}$. Note that $H^{\rm MIX}$ scales 
	up to logarithmic factor between $H^{\rm QUBO}$ and $H^{\rm HOBO}$. For 
	mixed Hamiltonian the constant $C$ satisfies $C\in(\frac{1}{2},1)$, and in 
	general depends on $\alpha$ and $N$.}
\end{table*}


\section*{Methods}

\paragraph{The analysis of circuits' depths}
Let us begin with HOBO and mixed approaches. According to 
Eq.~\eqref{eq:general-tsp-hobo} we can split all the terms into those defined 
over pairs $(b_t,b_{t'})$ for $t\neq t'$. For pairwise different 
$t_0,t_1,t_2,t_3$, 
if we have polynomials defined over $b_{t_0},b_{t_1}$, and $b_{t_2},b_{t_3}$, 
then we can implement them independently. Using round-robin schedule, we can 
implement those $\binom{N}{2}$ polynomials in $N-1$ ($N$) steps for even (odd) 
$N$, as it is described in Fig.~\ref{fig:schedule}. 

Let $ H$ be a general Hamiltonian defined over $K$ bits. If we implemented each 
term independently, then it would require $\sum_{k=1}^K 
2i\binom{K}{i}\Theta(2^KK) = 2^KK-1$ controlled NOTs according to the 
decomposition presented in Fig.~\ref{fig:decomposition-and-gray-code}a). Adding 
single auxiliary qubit and using the decomposition from 
Fig.~\ref{fig:decomposition-and-gray-code}b), and ordering terms according to 
Gray code, we can do it using $2^K$ qubits. Following the reasoning from 
previous paragraph we can apply only  $\lceil\frac{N}{2}\rceil$ Hamiltonians at 
once. We have an additional cost of $\lceil\frac{N}{2}\rceil$ qubits, however 
reducing the  depth cost by $K = \Theta (\log n)$ for both mixed and simple 
HOBO approaches 

As far as QUBO is concerned, we have to make separate analysis, since only 
2-local terms are present. Note that 1-local terms  $Z_{ti}$ can be implemented 
with circuit of depth 1. Moreover, for each $1\leq t \leq N$, terms 
$\{Z_{ti}Z_{tj}:1\leq i<j\leq N \}$ can be implemented with a circuit of depth 
$\approx N$ using round-robin schedule. We can similarly implement 
$\{Z_{ti}Z_{t'i}:1\leq t<t'\leq N \}$, which implement first two addends of 
$H^{\rm QUBO}$. For the last addend we have to implement for each $1\leq t \leq 
N$ terms $\mathcal Z = \{Z_{ti}Z_{t+1,j}: i \neq j \}$. First note that we can 
first implement them for even $t$, then for odd $t$, which doubles the depth of 
the circuit for single $t$. Finally, note that $Z = \bigcup_{k=0}^{n-1}\mathcal 
Z_k$, where each $ \mathcal Z_k= \{Z_{ti}Z_{t+1,i+k}|1\leq i \leq N\}$ can be 
implemented with circuit of depth 1. Eventually, the depth of Hamiltonian 
$H^{\rm QUBO}$ is of order $\Theta(N)$.

The detailed analysis for each encoding can be found in Supplementary materials.

\paragraph{Numerical analysis}

In order enable the simple reproduction of our results, we publish our code on GitHub \cite{adam_2020_4020610}. Below we present a detailed explanation of the optimization procedure. 

Let us describe the optimization algorithm we used to generate the result presented in Fig.~\ref{fig:numerics}. We have emulated the quantum evolution and take an exact expectation energy of the state during the optimization.
As a classical subroutine we used a L-BFGS algorithm implemented in Julia's Optim package \cite{Optim.jl-2018}. 
Independently of the instance of algorithm, we assume that the parameters 
$\theta_{\rm mix}$ for mixing Hamiltonian could be from the interval $[0,\pi]$. 
For objective Hamiltonian we assumed the parameter $\theta_{\rm obj}$ will be 
from $[0, R]$. For both we assume periodic domain, mainly if $\pi+\varepsilon$ 
($R+\varepsilon$) was encountered, the parameter was changed to $\varepsilon$, 
which changed the hypercube domain to hypertorus.

Let $r$ be the level numbers of the circuit. For $r< 5$, each run was started from randomly chosen vector within range of the parameters.
 For $r\geq 5$, we used a trajectories method similar to the one proposed in \cite{zhou2019quantum}. First we optimized the algorithm for $r=5$, as described in previous paragraph. Then, for each $r\geq 5$ we took the optimized parameter vectors $\vec \theta^{(r)}_{
\rm mix}$, $\vec \theta^{(r)}_{\rm obj}$ of length $r$, and constructed new vectors $\vec \theta^{(r+1)}_{\rm mix}$, $\vec \theta^{(r+1)}_{\rm obj}$ of length $r+1$ by copying the last element at the end. We took these vectors as initial points for $r+1$. Therefore we obtained a trajectory of length $11$ (6 for figure ~\ref{fig:numerics}f)) of locally optimal parameter vectors, one for each $r\geq 5$. 

Sometimes the algorithm has not converged to the local optimum in reasonable 
time. We claim that the reason came from periodicity of the domain, which for 
general TSP breaks the smoothness of the Hamiltonian landscape. We only 
accepted runs for which: for $r < 5$ the gradient was below $10^{-5}$; for 
$r\geq 5$, for each parameters vector from the trajectory, the gradient was  
supposed to be below $10^{-5}$. 

Since the energy values for both QUBO and HOBO are incomparable, we decided to 
present the probability of measuring the feasible solution, i.e. the solution 
describing a valid route.

Figures a,b,c from Figure~\ref{fig:numerics} were generated as follows. We took 
QUBO and HOBO encodings of TSP with $W\equiv 0$ for QAOA algorithm. One can 
consider it as a Hamiltonian problem on a complete graph. We took $A_1=A_2=1$ 
for both encodings, and $R=\pi$ ($R=2\pi$) for QUBO (HOBO). For each 
$r=1,\dots,15$ we generated 100 locally optimal parameter vectors, and for each 
$r$ we chose the maximum probability.

Figures d,e,f from Figure~\ref{fig:numerics} were generated as follows. We generated 100 matrices $W$ to be $W=(X+X^\top)$, where $X$ is a random matrix with i.i.d. elements from the uniform distribution over $[0,1]$. For each TSP instance we repeated the procedure as in $W\equiv 0$ case, however generating 40 samples for each $r$. For each TSP instance we chose the maximal probability of measuring the state in the feasible space. The line describes the mean of the best probabilities over all TSP instances. The area describes the range between the worst and the best cases of the best probabilities over all TSP instances.

\paragraph{The number of measurements}
For estimating the number of measurements we applied the Hoeffiding's inequality. Let $\bar X = \frac{1}{M}\sum_{i=1}^M X_i$ be the mean of i.i.d. random variables such that $X_i\in [a,b]$. Then 
\begin{equation}
\mathbb{P} (|\bar X - \mathbb E\bar X |  \geq t ) \leq 2\exp\left (- \frac{2M t^2}{b-a} \right ).
\end{equation}
In our case $\bar X$ is the energy estimation of the energy samples $X_i$. 
Provided that we expect both probability error and estimation error to be 
constant, we require $M = \Omega(b-a)$. 

The values of $a,b$ depend not only on the cost matrix $W$, but also on the 
form of the encoding and the values of constants $A_1,A_2,B$ in 
Eq.~\eqref{eq:general-tsp-hobo}. For simplicity, we take the following 
assumptions when estimating the samples number. First, w.l.o.g. we assume 
$B=1$. Furthermore, we assumed $C\max_{i\neq j} W_{ij}\leq A_i \leq 
C'\max_{i\neq j} W_{ij}$ where $C,C'$ do not depend on $N$ and $W$. This 
matches the minimal requirement for QUBO. While various measures on $W$ could 
be considered, we presented the results in the form $f(N) \max_{W_{ij}}(W)$, 
where $f$ does not depend on $W$. Furthermore, our analysis for each model is 
tight in $N$ assuming that $\max_{i\neq j}W_{ij},\min_{i\neq j}W_{ij} = 
\Theta(1)$ independently on $N$. Note that using this assumption $a\geq  
N\min_{i\neq j}W_{ij}$ is valid for any correctly chosen $A_1,A_2$.

Furthermore, Hamiltonians $H_{\neq}$ and $H_{\rm valid}$ are integer-valued, 
and the spectral gap is of constant order. For QUBO, the spectral gap is at 
most two, which can be generated by adding superfluous one-bit to any valid 
encoding. For HOBO, it can be generated by choosing the same number for 
different $b_t, b_{t'}$. Finally, for the mixed approach we can generate
incorrect assignments to slack $\xi_{t,i}$ variables. 
Theoretically, there is no upper-bound for $A_1,A_2$. However, in general it is 
not encouraged to make them too large, as classical optimization algorithm may 
focus too much on pushing the quantum state to feasible state instead of 
optimizing over feasible space. For this reason and to make the presentation of 
our results simpler, we assumed that $A_i$ are of order $\max_{i\neq j} 
W_{i,j}$.

During the optimization, one could expect that the quantum states will finally have large, expectedly $1-o(1)$ overlap with the feasible space. Thus one could expect that the estimated energy will be of typical, and later close to minimal route. Thus, for these points one could expect that $\order{N}\operatorname{range}(W)$ samples would be enough to estimate the energy accurately. We agree that it is a valid approach, especially when the gradient is calculated using $(f(\theta_0 +\varepsilon)-f(\theta_0 -\varepsilon))/(2\varepsilon)$ formula. However, recently a huge and justified effort has been made on analytical gradient estimation, which is calculated based on $\theta$ parameters that are far from the current $\theta_0$ point \cite{schuld2019evaluating,sweke2019stochastic}. In this scenario, we can no longer assume that the energy will be of the order of the typical route. Thus we believe that our approach for number of measurements estimation is justified.

The detailed analysis for each encoding can be found in the Supplementary Materials.


\paragraph{Acknowledgements}

AG has been partially supported by National Science
Center under grant agreement 2019/32/T/ST6/00158 and 2019/33/B/ST6/02011. AG would also
like to acknowledge START scholarship from the Foundation for Polish
Science. 
AK acknowledges financial support by the Foundation for Polish Science 
through TEAM-NET project (contract no. POIR.04.04.00-00-17C1/18-00).
ZZ acknowledges support from the Janos Bolyai Research Scholarship, the NKFIH Grants No. K124152, K124176, KH129601, K120569 and the Hungarian Quantum Technology National Excellence Program Project No. 2017-1.2.1-NKP-2017-00001.
This research was supported in part by PL-Grid Infrastructure.

\bibliographystyle{ieeetr}
\bibliography{hobo_tsp_vqe_notes}
\clearpage
\appendix
\onecolumn 

\section{Supplementary analysis}

\subsection{Model resources analysis}
In this subsection we will provide a detailed analysis of resources required for each model. 

\subsubsection{QUBO}
Let us recall that the QUBO model takes the form
\begin{equation}
H^{\rm QUBO}(b) = A_1 \sum_{t=1}^N \left( 1 -\sum_{i=1}^N b_{ti} \right)^2 + A_2 \sum_{i=1}^N \left( 1 -\sum_{t=1}^N b_{ti} \right)^2 + B\sum_{\substack{i,j=1\\i\neq j}}^N W_{ij} \sum_{t=1}^N b_{ti}b_{t+1,j} \label{eq_sup:qubo}
\end{equation}
\paragraph{Number of qubits} The model requires  $N^2$ qubits. 

\paragraph{Number of terms} The number of terms can be determined as follows. First we note that we have $N^2$ 1-local terms. Secondly, from the first addend for each $t$ we have $\binom{N}{2}$ 2-local terms, similarly for the second. Finally, for the last part for each $i\neq j$ we have additional $N$ 2-local terms. Note that each 2-local term is present only in one part, which makes our calculation tight. Finally we have
\begin{equation}
\#{\rm terms} = N^2 + 1 + 2N\binom{N}{2} + NN(N-1) = 2N^3 - N^2 + 1.
\end{equation}

\paragraph{Depth of the circuit} Following the reasoning presented in the Method section, we can conclude that the 1-local terms can be implemented with the circuit of depth 1. The first addend from Eq.~\eqref{eq_sup:qubo} can be implemented with the circuit of depth $N$ for even $N$, and $N-1$ for odd $N$, counting $Z_iZ_j$ gates. 
The second addend can be studied analogously. 

For last addend we can independently consider parts 
$\sum_{i\neq j} W_{ij} b_{ti}b_{t+1,j}$ for even $t$, and then for odd $t$, which will double the circuit depth comparing to fixed $t$. Let us fix $t$. We can implement terms $ \mathcal Z_k= \{Z_{t,i}Z_{t+1,i+k}|1\leq i \leq N\}$ with circuit of depth 1.  Since $1\leq k \leq N$, we can simulate the last addend with the circuit of depth $2N$. Thus in total our circuit has depth $4N+1$

Since we have $N^2$ qubits, we can simulate at most $N^2/2$ 2-local terms 
independently. We have $\sim 2N^3$ terms and $N^2$ qubits. The circuit depth 
for simulating 2-local terms is $\sim 2N^3 / (N^2/2)=4N$, which shows that our 
analysis is tight.

The calculations were done in terms of gates $\exp(-\ii t Z)$ and $\exp(-\ii t 
Z_iZ_j)$. Since the latter requires 2 CNOTs and a single rotation gate, we have 
to triple the circuit depth implementing 2-local gates, which finally gives us 
$12N+1$.

\paragraph{Number of measurements} For the sake of simplicity we assume that $A_1,A_2\leq C\max_{i\neq j}W_{ij}$ and $B=1$. 
Note that for each $t$, the expression $\left(1-\sum_{i=1}^Nb_{ti}\right)^2$ 
can be bounded from above by $(N-1)^2$. We can similarly upperbound the 
next addend. For the last part we have
\begin{equation}
\sum_{\substack{i,j=1\\i\neq j}}^N W_{ij} \sum_{t=1}^N b_{ti}b_{t+1,j} \leq N\sum_{\substack{i,j=1\\i\neq j}}^N W_{ij} \leq N\binom{N}{2} \max_{i\neq j}W_{ij}
\end{equation}
Finally we have
\begin{equation}
\begin{split}
H^{\rm QUBO}(b) &\leq A_1 N(N-1)^2 + A_2 N(N-1)^2 + BN\binom{N}{2} \max_{i\neq j}W_{ij} \\
&\leq C N^3 \max_{i\neq j}W_{ij} + CN^3\max_{i\neq j}W_{ij} + N^3 \max_{i\neq j}W_{ij}\\
&= C' N^3 \max_{i\neq j}W_{ij}
\end{split}
\end{equation}
Note that the results is tight in order of $N$, which can be shown using $b_{ti}\equiv 1$ assignment.

\subsubsection{HOBO}
Let us recall that the model takes the form
\begin{equation}\label{eq_sup:hobo}
\begin{split}
H^{\rm HOBO}(b) &= A_1\sum_{t=1}^N H^{\rm HOBO}_{\rm valid} (b_t) + A_2\sum_{t=1}^N \sum_{t'=t+1}^N H^{\rm HOBO}_{\neq}(b_t,b_{t'}) \\
&\phantom{\ =}+ B\sum_{\substack{i,j=1\\i\neq j}}^N W_{ij} \sum_{t=1}^N H^{\rm 
HOBO}_\delta (b_t,i)H^{\rm HOBO}_\delta (b_{t+1},j).
\end{split}
\end{equation}
Provided that $\tilde b_{K-1}\dots \tilde b_{0}$ is a binary representation of 
$N-1$, we define
\begin{gather}
H_{\rm valid}^{\rm HOBO}(b_t) \coloneqq \sum_{k^0\in K_0} b_{t,k_0}\prod_{k=k^0+1}^{K-1} (1-(b_{t,k} - \tilde b_{k})^2)\\
H_{\neq }^{\rm HOBO}(b,b')\coloneqq H^{\rm HOBO}_\delta(b,b') \coloneqq 
\prod_{k=1}^K (1- 
(b_k - b'_k)^2)
\end{gather}
where $k^0 \in K_0$ are such indices that $\tilde{b}_{k^0} = 0$.

The proof that $H_{\rm valid}^{\rm HOBO}$ is a valid Hamiltonian for this encoding is presented in Sec.~\ref{sec_sup:h_valid_hobo}. For the sake of convenience, we will assume $K=\lceil\log(N) \rceil$, which is at the same time the number of bits in $b_t$ for any $t$.

\paragraph{Number of qubits} The required number of qubits is $N K +\frac{N}{2} 
\sim N\log(N)$. The $\frac{N}{2}$ part comes from the Gray code technique 
presented in the main text.

\paragraph{Number of terms} We will assume that $H_{\neq }^{\rm HOBO }(b_t,b_{t'})$ is a Hamiltonian consisting of all terms, meaning that for any factor of this Hamiltonian, the corresponding $\alpha_S$ is non-zero. While theoretically some terms may vanish, our investigations showed that at least for $N=3,4,5$ it is not the case. Furthermore, this guarantees that the bound will be valid for the number of terms in corresponding Ising model.

Due to this assumption, we do not need to consider $H_{\rm valid}^{\rm HOBO}$ 
anymore, since terms in $H_{\rm valid}^{\rm HOBO}(b_t)$ are present in  
$H_{\neq }^{\rm HOBO}(b_{t},b_{t+1})$ as well. For the same reason we do not 
have to consider elements from $H^{\rm HOBO}_\delta (b_t,i)H^{\rm HOBO}_\delta 
(b_{t+1},j)$.

Note that considering for example $H_{\neq }^{\rm HOBO}(b_{t},b_{t+1})$ and $H_{\neq }^{\rm HOBO}(b_{t+1},b_{t+2})$ we doubly count the terms defined over $b_{t}$ only. In general, each term will be counted $N-1$-times. Taking all of this into account we obtain
\begin{equation}
\begin{split}
\# {\rm terms} &= \binom{N}{2} 2^{2K} - (N-2)N2^K \leq \frac{1}{2}N^22^{2\log N} - N (N-2)2^{\log N-1} = \frac{1}{2}N^4 - \frac{1}{2}N^2(N-2) = \frac{1}{2}N^4 - \frac{1}{2}N^3 + N^2. 
\end{split}
\end{equation}

\paragraph{Depth of the circuit} Similarly as it was done while counting terms, we will only consider implementing $H_{\neq }^{\rm HOBO}$, as by commutativity $\alpha Z + \beta Z = (\alpha + \beta)Z$. 

Following the round-robin schedule we can implement $H_{\neq }^{\rm HOBO}$ in pairs in at most $N$ steps. Following the Gray code decomposition presented in the main text, we can implement each $H_{\neq}^{\rm HOBO}$ using a circuit of depth $2\cdot2^{2K}-1$. This gives us a final depth
\begin{equation}
{\rm depth} \leq N 2\cdot2^{2K}-1 \leq 2N \cdot 2^{2\log N} -1  = 2N^3 - 1.
\end{equation}

Note that factors defined over single $b_t$ will be implemented several times. This conflict can be solved by applying the corresponding angle only once, and for the rest of occurrences applying 0-angle rotation. Since we are not applying rotation in that case, we can reduce the circuit by up to $(N-1)2^{K} \approx N^2$, which does not have significant impact on the formula derived above.

It is complicate to reasonable lower bound the circuit's depth. Since most of 
the factors are of order $\log N$, one could consider that applying each term 
requires the depth of the same order as well. However using Gray code ordering 
it is clear that only two quits may be needed for applying higher-local terms. 
For this reason we will assume that at each step only two qubits are required 
to implement each term. This gives us the lower bound $\approx\frac{1}{2}N^4 / 
(N \log(N)) = \frac{1}{2}N^3/\log N$ which shows that our approach is tight up 
to $\log(N)$ factor.

\paragraph{Number of measurements}
For simplicity, we will assume that $A_1,A_2\leq C\max_{i\neq j}W_{ij}$ and 
$B=1$. Note that $H_{\rm valid}^{\rm HOBO}$ is a sum of at most $K-1$ elements, 
each giving the value either 0 or 1. Hence, for each $t$ we have $H_{\rm 
valid}^{\rm HOBO}(b_t)\leq K-1$.

Note that $H_{\neq}^{\rm HOBO} \equiv H_{\delta}^{\rm HOBO}$ and 
$H_{\delta}^{\rm HOBO}(\cdot, \cdot )\in \{0,1\}$. 
Furthermore, since $b_t$ can decode a single number only, $H_{\delta}(b_t,i)=1$ 
only for a single $i$. Thus
\begin{equation}
\begin{split}
\sum_{\substack{i,j=1\\i\neq j}}^N W_{ij} \sum_{t=1}^N H^{\rm HOBO}_\delta (b_t,i)H^{\rm HOBO}_\delta (b_{t+1},j) &= \sum_{t=1}^N \sum_{\substack{i,j=1\\i\neq j}}^N W_{ij} H^{\rm HOBO}_\delta (b_t,i)H^{\rm HOBO}_\delta (b_{t+1},j) \\
&= \sum_{t=1}^N W_{b_t,b_{t+1}} \leq N \max_{i\neq j} W_{ij}.
\end{split}
\end{equation}
and we can upper bound the energy by 
\begin{equation}
\begin{split}
H^{\rm HOBO}(b) &\leq A_1 N (K-1) + A_2 \binom{N}{2} + BN \max_{i\neq j} W_{ij}\\
&\leq C N \log N \max_{i\neq j} W_{ij} + C \frac{N^2}{2} \max_{i\neq j} W_{ij} 
+ N \max_{i\neq j} W_{ij} \\
&\leq C' N^2 \max_{i\neq j} W_{ij}.
\end{split}
\end{equation}
One can again shown the tightness of the bound by using $b_{ti} \equiv 1$.

\subsubsection{Mixed approach}
The Hamiltonian takes the form
\begin{equation}
\begin{split}
H^{\rm MIX}(b) &= A_1\sum_{t=1}^N H^{\rm MIX}_{\rm valid} (b_t; \xi_t) + A_2\sum_{t=1}^N \sum_{t'=t+1}^N H_{\neq}^{\rm MIX}(b_t,b_{t'}) \\
&\phantom{\ =}+ B\sum_{\substack{i,j=1\\i\neq j}}^N W_{ij} \sum_{t=1}^N H^{\rm MIX}_\delta (b_t,i)H^{\rm MIX}_\delta (b_{t+1},j),
\end{split}
\end{equation}
where $\xi_t$ are slack variables required to implement $H_{\rm valid}^{\rm MIX}$ and
\begin{gather}
H_{\rm valid}^{\rm MIX}(b_t) \coloneqq \left ( -\sum _{l=1}^L \sum_{k=0}^{K-1} b_{tlk} +1 + \sum_{i=0}^{\lceil\log (KL)\rceil }2^i \xi_{t,i} \right )^2 +  \sum _{l=1}^L \left ( \sum_{k=0}^{K-1} b_{tlk} \right )\left ( \sum_{\substack{l'=1 \\ l'\neq l}}^L\sum_{k=0}^{K-1} b_{tlk} \right )\\
H_{\neq}^{\rm MIX}(b_t,b_{t'}) \coloneqq \sum_{l=1}^L \left( \sum_{k=0}^{K-1}(b_{tlk}+b_{t',l,k})\right)\prod_{k=1}^K (1- (b_{t',l,k} - b'_{t,l,k})^2)  \\
H_\delta(b_t,i)^{\rm MIX} \coloneqq \prod_{k=1}^K (1- (b_{t',\bar l(i),k} - b^i_{k})^2).
\end{gather}

For a general choice of $\alpha\in (0,1)$ it is hardly possible that $\alpha 
\log N$ will be an integer number. Hence for fixed $\alpha$ let $K\coloneqq 
\lfloor \alpha \log N\rfloor$. Note that $K\sim \alpha \log N$. On the other 
hand, we will encounter elements of the form $2^K$ and $2^{2K}$, for which such 
an equivalence is not always valid, as
\begin{equation}
2^{K} = 2^{\lfloor \alpha \log N\rfloor} = 2^{\alpha \log N - \varepsilon_\alpha(N)} = C_\alpha(N)N^\alpha,
\end{equation}
where $C_\alpha(N) \coloneqq 2^{- \varepsilon_\alpha(N)}$ depends on the choice of $\alpha$ and $N$, but always $\frac{1}{2} \leq C_\alpha(N) \leq 1$. Similarly
\begin{equation}
2^{2K}=2^{2\lfloor \alpha \log N\rfloor} = 2^{2\alpha \log N - 2\varepsilon_\alpha(N)} = C_\alpha^2(N)N^{2\alpha}.
\end{equation}

Furthermore
\begin{equation}
L \coloneqq \left \lceil \frac{N}{2^K-1} \right\rceil  \sim
\left \lceil \frac{N}{C_\alpha(N)N^\alpha} \right\rceil \sim\frac{1}{C_\alpha(N)}N^{1-\alpha}
\end{equation} 
Note that if $N\neq (2^{K}-1)L$, then we have to add a separate Hamiltonian of 
the form similar to $H_{\rm valid}^{\rm HOBO}$, as this encoding will not a 
encode valid city. This does not change the estimations derived in next 
paragraphs, as 
\begin{itemize}
	\item it does not require additional qubits,
	\item it does not produce new terms (they are already included in 
	$H_{\neq}^{\rm HOBO})$, and by this it does not change the depth of the 
	circuit,
	\item it has negligible impact on the energy upperbound, since for each $t$ the mentioned Hamiltonian will increase energy by at most $K$. 
\end{itemize}


\paragraph{Number of qubits} The Hamiltonian requires 
\begin{equation}
\begin{split}
	NKL + \left \lfloor \frac{N}{2}\right \rfloor L + N \lceil \log(KL) \rceil &\sim \frac{\alpha}{C_\alpha(N)} N N^{1-\alpha} \log N + \frac{1}{2C_\alpha(N)}N N^{1-\alpha} +N \left (1+\log ( \frac{\alpha}{C_\alpha(N)} N^{1-\alpha} \log N)\right)\\
	& = \frac{\alpha}{C_\alpha(N)} N^{2-\alpha}\log N + \frac{1}{2C_\alpha(N)}N^{2-\alpha}  + N \operatorname{poly}(\log(N))
\end{split}
\end{equation}
qubits.
The $\left \lfloor \frac{N}{2}\right \rfloor L$ is required for implementing the scheduling, while $N \lceil \log(KL) \rceil$ qubits are needed for $\xi$ variables.

\paragraph{Number of terms}
	Let us start by calculating the terms generated from $H_{\neq }^{\rm MIX}$. 
	Let us fix a pair of different time points $t,t'$. The Hamiltonian consists 
	of $L$ independent Hamiltonians defined over $b_{t,l},b_{t,l'}$, which 
	consist of $2^{2K}$ elements. As it was in the case of $H_{\neq}^{\rm 
	MIX}$, we were over-counting terms defined over $b_{t,l}$ only. Taking this 
	into account we have
	\begin{equation}
	\binom{N}{2}L2^{2K} - NL(N-2)2^K \sim \frac{C_\alpha(N)}{2}N^2 N^{1-\alpha} N^{2\alpha} - N^2 N^{1-\alpha} N^\alpha = 
	\frac{C_\alpha(N)}{2}N^{3+\alpha} - N^{3}
	\end{equation}
	terms. 
	
	Note that in the case of $H_{\delta}^{\rm MIX}$ new factors appear for some $i,j$ with $\bar l(i)\neq \bar l(j)$. They still are Hamiltonians defined over two binary vectors $b_{t,l(i)},b_{t+1,l(j)}$, each generating $2^{2K}$ terms. Note that if for some $i',i$ we have $\bar l(i)=\bar l(i')$, then Hamiltonians 
	\begin{equation}
	\prod_{k=1}^K (1- (b_{t',\bar l(i),k} - b^i_{k})^2)\prod_{k=1}^K (1- (b_{t',\bar l(j),k} - b^j_{k})^2)
	\end{equation}
	and
	\begin{equation}
	\prod_{k=1}^K (1- (b_{t',\bar l(i'),k} - b^{i'}_{k})^2)\prod_{k=1}^K (1- (b_{t',\bar l(j),k} - b^j_{k})^2)
	\end{equation}
	are defined over the same variables, thus we only have to consider terms over different $l$ instead of different $i$.
	 Taking all of this plus over-counted terms into account, we obtain
	\begin{equation}
	NL(L-1)2^{2K}- NL(2L-2)2^{K} = NL^2 (2^{2K}- 2\cdot 2^K) \sim  N^3- \frac{2}{C_\alpha(N)}N^{3-\alpha}.
	\end{equation}

	Finally, let us consider $H_{\rm valid}^{\rm MIX}$. Note that the terms 
	produced by the latter addend were already considered. Hence we need to 
	consider the only terms from the first addend. Let $t$ be fixed. The terms 
	defined over $b$ only were already considered.  We have $\lceil \log 
	(KL)\rceil $ 1-local terms $\xi_{t,i}$. We have $KL\lceil \log (KL)\rceil$ 
	2-local terms of the form $b_{tlk}\xi_{t,i}$. Finally, we have 
	$\binom{\lceil \log (KL)\rceil}{2}$ terms of the form $\xi_{t,i}\xi_{t,j}$. 
	Taking all of this into account, for each $t$ we have
	\begin{equation}
	\lceil \log (KL)\rceil + KL\lceil \log (KL)\rceil + \binom{\lceil \log (KL)\rceil}{2} \sim (1-\alpha)\log N + \frac{\alpha(1-\alpha)}{C_\alpha(N)}N^{1-\alpha}\log^2 N + \frac{(1-\alpha)^2}{2}\log^2(N).
	\end{equation}
	
	Taking all numbers above we see that
	\begin{equation}
	\# {\rm terms} \sim \frac{C_\alpha (N)}{2} N^{3+\alpha}. 
	\end{equation}
	
\paragraph{Depth of the circuit}

Let us begin with the terms introduced by $H_{\neq}^{\rm MIX}$.
The Hamiltonian can be splitted into Hamiltonians defined over $b_{t,l},b_{t',l}$. Note that $l$ are always shared, thus for different $l$, the circuits can be implemented independently. For fixed $l$ we can use the same techniques as for $H_{\neq}^{\rm HOBO}$, which gives us the depth
\begin{equation}
N(2\cdot 2^{2K}-1) = 2 C_\alpha^2(N) N^{1+2\alpha} - N.
\end{equation}

Similarly, the terms introduced by $H_{\delta}^{\rm MIX}$ and not considered before are Hamiltonians defined over $b_{t,l},b_{t+1,l'}$ with $l\neq l'$. In this case we can independently implement Hamiltonians first for even $t$, then for odd $t$, which will double the circuit's depth. For fixed $t$ we can use similar technique of arranging the Hamiltonians as it was for the last addend of $H^{QUBO}$. This gives us the depth
\begin{equation}
2(L-1)(2\cdot 2^{2K}-1) \sim 4C_\alpha(N) N^{1+\alpha}.
\end{equation}

In the case of $H_\delta^{\rm HOBO}$ and $H_{\neq}^{\rm HOBO}$ we have not taken into account the overcounted elements defined over $b_{t,l}$ only. However, at it was in the case of $H_{\neq}^{\rm HOBO}$, they do not change the leading factor.

Elements from $H_{\rm valid}^{\rm MIX}$ can be implemented using a circuit of depth $3(N+\log(KL))+1$, since for each $t$, Hamiltonians can be applied independently, and each $\xi_{t,i}$ has to be applied with all $b_{ti}$ and $b_{t,j}$. Part $+1$ comes from applying 1-local terms $\xi_{t,i}$.

Taking all of the numbers derived above we finally have obtain
$\sim 2C^2_\alpha(N) N^{1+2\alpha}$.

Similarly as it was in the case of $H^{\rm HOBO}$, the minimum depth in case of mixed approach is $\frac{C_\alpha (N)}{2}N^{3+\alpha}/(\frac{\alpha}{C_\alpha (N)} N^{2-\alpha}\log N)= \frac{C_\alpha^2(N)}{2\alpha}N^{1+2\alpha}/\log N$.

\paragraph{Number of measurements}
For the sake simplicity, we will assume that $A_1,A_2\leq C\max_{i\neq j}W_{ij}$ and $B=1$. 
For general $b$ we have
\begin{equation}
\begin{split}
H(b) &\leq A_1N(LK-1)^2 + A_1N\cdot L\cdot K\cdot LK + A_2 \binom{N}{2} L \cdot 2K +BN \max_{i \neq j}W_{ij} \\
&\leq \left (2CNL^2K^2 + CN^2LK  +N\right )\max_{i \neq j}W_{ij}   \\
& \sim \left (\frac{2C\alpha^2}{C^2_\alpha (N)}N^{3-2\alpha}\log^2 N + CN^{3-\alpha}\log N  +N\right )\max_{i \neq j}W_{ij}.
\end{split}
\end{equation}
By this we conclude that $H(b) =\order{N^{3-\alpha}\log N}\max_{i \neq j}W_{ij}$. Note that the bound is achievable when taking $b_{tlk}\equiv 1$. 

\subsection{Proof for $H_{\rm valid}^{\rm HOBO}$}\label{sec_sup:h_valid_hobo}

\begin{theorem}
	Let $N>0$ and $K$ satisfies $2^{K-1}\leq N<2^K$.
	Let $\tilde b= \tilde b_{K-1}\dots \tilde b_{0}$ is a binary representation of $N-1$. Let $K_0\subseteq \{0,\dots, K-1\}$ be indices such that $k_0\in K_0$ iff $\tilde b_{k_0}=0$ Let
	\begin{equation}
	H(b) \coloneqq \sum_{k_0\in K_0} b_{k_0}\prod_{k=k_0+1}^{K-1} (1-(b_{k} - \tilde b_{k})^2)
	\end{equation}
	and $b= b_{K-1}\dots  b_{0}$ be a vector of bits encoding some number $n\in \{0,\dots,2^K-1\}$. Then $H(b)\geq0$, with equality iff $n<N$.
\end{theorem}
\begin{proof}
	Note that $(1-(b_{k} - \tilde b_{k})^2)$ is nonnegative, hence $H(b)\geq0$ 
	independently on $b$. Let $n= N-1$, which means $b=\tilde b$. Then
	\begin{equation}
	H(\tilde b) = \sum_{k_0\in K_0} \tilde b_{k_0}\prod_{k=k_0+1}^{K-1} 
	(1-(\tilde b_{k} - \tilde b_{k})^2) = \sum_{k_0\in K_0}  0\cdot 
	\prod_{k=k_0+1}^{K-1} (1-(\tilde b_{k} - \tilde b_{k})^2)  = 0.
	\end{equation}
	
	Let $n< N-1$. Then there exists a unique 
	$k'\in\{0,\dots, K-1\}\setminus K_0$ 
	such that for all $k>k'$ we have $b_{k} = \tilde b_k $,  $b_{k'} = 0$. In 
	other words, there exists a bit, which for $N-1$ is one, and for 
	$n$ is 0. It is the first one starting from most 
	significant one. Then we have
	\begin{equation}
	\begin{split}
	H( b) &=\sum_{k_0\in K_0} b_{k_0}\prod_{k=k_0+1}^{K-1} (1-(b_{k} - \tilde b_{k})^2) \\
	&= \sum_{\substack{k_0\in K_0\\k_0>k'}} b_{k_0}\prod_{k=k_0+1}^{K-1} (1-(b_{k} - \tilde b_{k})^2) +
	\sum_{\substack{k_0\in K_0\\k_0<k'}} b_{k_0}\prod_{k=k_0+1}^{K-1} (1-(b_{k} - \tilde b_{k})^2)\\
	&=  \sum_{\substack{k_0\in K_0\\k_0>k'}} \tilde b_{k_0}\prod_{k=k_0+1}^{K-1} (1-(\tilde b_{k} - \tilde b_{k})^2) +
	\sum_{\substack{k_0\in K_0\\k_0<k'}} b_{k_0}\prod_{k=k_0+1}^{K-1} (1-(b_{k} - \tilde b_{k})^2)\\
	&=  \sum_{\substack{k_0\in K_0\\k_0>k'}} 0\cdot \prod_{k=k_0+1}^{K-1} (1-(\tilde b_{k} - \tilde b_{k})^2) +
	\sum_{\substack{k_0\in K_0\\k_0<k'}} b_{k_0}(1-(b_{k'} - \tilde b_{k'})^2)\prod_{\substack{k=k_0+1\\k\neq k'}}^{K-1} (1-(b_{k} - \tilde b_{k})^2) \\
	&=0+ \sum_{\substack{k_0\in K_0\\k_0<k'}} b_{k_0}(1-(0 - 
	1)^2)\prod_{\substack{k=k_0+1\\k\neq k'}}^{K-1} (1-(b_{k} - \tilde 
	b_{k})^2) = 0.
	\end{split}
	\end{equation}
	
	Let $n>N$. Then there exists a unique $k'\in K_0$ such that for all $k>k'$ 
	we have $b_{k} = \tilde b_k $ and $b_{k'} = 1$. In other words, there 
	exists a bit, which for bit from $N-1$ is zero, and for bit from $n$ 
	is one It is the first one starting from most significant 
	one. Then, taking the addend to $k_0=k'$ we have
	\begin{equation}
	b_{k'}\prod_{k=k'+1}^{K-1} (1-(b_{k} - \tilde b_{k})^2) = 1 \prod_{k=k'+1}^{K-1} (1-(\tilde b_{k} - \tilde b_{k})^2) = 1,
	\end{equation}
	which  is enough to prove that $H(b)>0$ as each addend is nonnegative.
\end{proof}
\end{document}